\documentclass[sigconf,nonacm]{acmart}

\usepackage{amsthm, amsfonts, amsmath} 
\usepackage{graphicx}
\usepackage{appendix}
\usepackage{algorithm}
\usepackage{algorithmicx}
\usepackage{algpseudocode}
\usepackage{array} 
\usepackage{booktabs} 
\usepackage{tabulary} 
\usepackage{multirow} 
\usepackage[flushleft]{threeparttable} 
\usepackage{xspace}
\usepackage{enumerate}
\usepackage[disable, colorinlistoftodos,prependcaption,textsize=tiny]{todonotes}
\usepackage[font=small]{caption}
\usepackage{mathtools}
\usepackage{subcaption}

\newtheorem{theorem}{Theorem}[section]
\newtheorem{definition}[theorem]{Definition}
\newtheorem{remark}[theorem]{Remark}

\newtheorem{lemma}[theorem]{Lemma}
\newtheorem{corollary}[theorem]{Corollary}
\newtheorem{conjecture}[theorem]{Corollary}

\newcommand{\bigO}{\smash{\ensuremath{O}}}
\newcommand{\tilO}{\smash{\ensuremath{\widetilde{O}}}}
\newcommand{\tilOm}{\smash{\ensuremath{\widetilde{\Omega}}}}
\newcommand{\tilT}{\smash{\ensuremath{\widetilde{\Theta}}}}

\newcommand{\E}{\mathbb{E}}
\newcommand{\eps}{\varepsilon}
\newcommand{\calA}{\mathcal{A}}
\newcommand{\calB}{\mathcal{B}}

\newcommand{\calG}{\mathcal{G}}

\newcommand{\code}[1]{\texttt{#1}}

\DeclareMathOperator{\polylog}{polylog}
\DeclareMathOperator{\hop}{hop}
\DeclareMathOperator*{\argmin}{arg\,min}

\newcommand{\hybrid}{\ensuremath{\mathsf{HYBRID}}\xspace}

\newcommand{\hybridpar}[2]{\ensuremath{\mathsf{HYBRID}(#1,#2)}}
\newcommand{\hybridparbig}[2]{\ensuremath{\mathsf{HYBRID}\big(#1,#2\big)}}
\newcommand{\LOCAL}{\ensuremath{\mathsf{LOCAL}}\xspace}
\newcommand{\CONGEST}{\ensuremath{\mathsf{CONGEST}}\xspace}

\newcommand{\NCC}{\ensuremath{\mathsf{NCC}}\xspace}
\newcommand{\CC}{\ensuremath{\mathsf{CLIQUE}}\xspace}

\newif\iflong
\longtrue

\iflong\parskip = 0.3\baselineskip
\fi

\newcommand{\lng}[1]{\iflong#1\fi}
\newcommand{\shrt}[1]{\iflong\else#1\fi}

\begin{document}

\title{Routing Schemes and Distance Oracles in the Hybrid Model}

\author{Fabian Kuhn and Philipp Schneider} 

\email{kuhn@cs.uni-freiburg.de, philipp.schneider@cs.uni-freiburg.de}
%
%

\begin{abstract}
The $\mathsf{HYBRID}$ model was introduced as a means for theoretical study of \textit{distributed} networks that use various communication modes. Conceptually, it is a synchronous message passing model with a \textit{local communication mode}, where in each round each node can send large messages to all its neighbors in a local network (a graph), and a \textit{global communication mode}, where each node is allotted limited (polylogarithmic) bandwidth per round which it can use to communicate with \textit{any} node in the network.

Prior work has often focused on shortest paths problems in the local network, as their global nature makes these an interesting case study how combining communication modes in the \hybrid model can overcome the individual lower bounds of either mode. In this work we consider a similar problem, namely computation of \textit{distance oracles} and \textit{routing schemes}. In the former, all nodes have to compute \textit{local tables}, which allows them to look up the distance (estimates) to any target node in the local network when provided with the \textit{label} of the target. In the latter, it suffices that nodes give the next node on an (approximately) shortest path to the target.

Our goal is to compute these local tables as fast as possible with labels as small as possible. We show that this can be done \textit{exactly} in $\widetilde O(n^{1/3})$ communication rounds and labels of size $\Theta(n^{2/3})$ bits. For constant stretch approximations we achieve labels of size $O(\log n)$ in the same time. Further, as our main technical contribution, we provide computational lower bounds for a variety of problem parameters. For instance, we show that computing solutions with stretch below a certain constant takes $\widetilde \Omega(n^{1/3})$ rounds even for labels of size $O(n^{2/3})$.
\end{abstract}

\maketitle

\section{Introduction}

\paragraph{Hybrid Networks} Real networks often employ multiple communication modes. For instance, mobile devices combine high-bandwidth, short-range wireless communication with relatively low-bandwidth cellular communication (c.f., 5G \cite{Asadi2016}). Other examples are hybrid data centers, which combine wireless and wired communication \cite{Han2015} or optical circuit switching and electrical packet switching \cite{Wang2010}, or software defined networking \cite{Vissicchio2014}.

In this article we utilize the theoretical abstraction of such hybrid communication networks provided by \cite{Augustine2020} which became known as \textit{hybrid model} and was designed to reflect a high-bandwidth local communication mode and a low-bandwidth global communication mode, capturing one of the main aspects of real hybrid networks.
Fundamentally, the hybrid model builds on the concept of \textit{synchronous message passing}, a classic model to investigate communication complexity in distributed systems. 

\begin{definition}[Synchronous Message Passing, c.f., \cite{Lynch1996}] 
	\label{def:sync_msg_passing}
	Let $V$ be a set of $n$ nodes with unique identifiers $ID)(\cdot):V \to [n]$.\footnote{$[k] := \{1, \ldots , k\}$ for $k \in \mathbb{N}$.} Time is slotted into discrete rounds consisting of the following steps. First, all nodes receive the set of messages addressed to them in the last round. Second, nodes conduct computations based on their current state and the set of received messages to compute their new state (randomized algorithms also include the result of some random function). Third, based on the new state the next set of messages is sent.
\end{definition}

Synchronous message passing has a clear focus on investigating communication complexity, i.e., the number of communication rounds required to solve a problem with an input distributed over all nodes (usually a graph problem). For this purpose, nodes are usually assumed to be computationally unbounded.\footnote{Occasionally this model is ``overexploited'', e.g., nodes are supposed solve $\mathcal{NP}$-complete problems on their local data. We will refrain from that.} The hybrid model then places additional restrictions on the messages size and which pairs of nodes can exchange them.

\begin{definition}[Hybrid model \cite{Augustine2020}] 
	\label{def:hybrid_model}
	The \hybridpar{\lambda}{\gamma} model is a synchronous message passing model (Def.\ \ref{def:sync_msg_passing}), subject to the following restrictions. \emph{Local mode:} nodes may send one message per round of maximum size $\lambda$ bits to each of their neighbors in a graph. \emph{Global mode:} nodes can send and receive messages of total size at most $\gamma$ bits per round to/from any other node(s) in the network. If the restrictions are not adhered to then a strong adversary\footnote{The strong adversary knows the states of all nodes, their source codes and even the outcome of all random functions.} selects the messages that are delivered.
\end{definition}

Note that the parameter spectrum of the \hybridpar{\lambda}{\gamma} model covers the standard models \LOCAL, \CONGEST, \CC (aka ``Congested Clique'') and \NCC (``Node Capacitated Clique'') as marginal cases.\footnote{\LOCAL: $\lambda \!=\! \infty, \gamma \!=\! 0$, \CONGEST: $\lambda \!=\! \bigO(\log n), \gamma \!=\! 0$, \CC (+ Lenzen's Routing scheme \cite{Lenzen2013}): $\lambda \!=\! 0, \gamma \!=\! n \log n$, \NCC: $\lambda \!=\! 0, \gamma \!=\! \bigO(\log^2 n)$.} Given the ramifications of investigating \hybridpar{\lambda}{\gamma} in its entirety, we narrow our scope (for our upper bounds) to a particular parametrization that pushes both communication modes to one extreme end of the spectrum. Following the argumentation of \cite{Augustine2020} we leave the size of local messages unrestricted (modeling high local bandwidth) and allow only $\polylog n$ bits of global communication per node per round (modeling severely restricted global bandwidth). Formally, we define the ``standard'' hybrid model as combination of the standard \LOCAL and \NCC \cite{Augustine2019} models: $\hybrid := \hybridparbig{\infty}{\bigO(\log^2 n)}$.\footnote{Choosing $\gamma = \bigO(\log^2 n)$ is due to convenience. It allows nodes to exchange $\bigO(\log n)$ global messages of size $\bigO(\log n)$ bits, which often makes randomized algorithms more concise. The slight arbitrariness of this choice is one reason to resort to the $\tilO$ notation for analyzing the round complexity.} Note that our lower bounds are parametrized for the more general $\hybridpar{\infty}{\gamma}$ model (which also yields lower bounds for the weaker \hybrid model).

\paragraph{Routing Schemes \& Distance Oracles} A fundamental aspect of the Internet Protocol is packet forwarding, where every node has to compute a routing function, which -- when combined with target-specific information stored in the packet header -- must indicate the neighbor which the packet has to be forwarded to such that it reaches its intended destination. A correct \textit{routing scheme} consists of these routing functions and a unique \textit{label} per node, such that the packet forwarding procedure induces a path in the network from \textit{any} source node to \textit{any} destination node specified by the corresponding label attached to the packet. 

Typically, a distinction is made between \textit{stateful} and \textit{stateless} routing schemes. In the former, routing can be based on additional information accumulated in the packet header as the packet is forwarded, whereas in the latter routing decisions are completely oblivious to the previous routing path. A related problem is the computation of distance oracles, which has some similarities to the all pairs shortest paths problem. Each node must compute an oracle function that provides the distance (or an estimate) to any other node when provided with the corresponding label. Formal definitions are given in Section \ref{sec:preliminaries} (Def.\ \ref{def:distance_oracles}, \ref{def:stateless_routing} and \ref{def:stateful_routing}).

Our first goal is to gather the necessary information for labels, routing and oracle functions with as few communication rounds as possible. This is particularly important for dynamic or unreliable networks where changes in distances or topology necessitates (frequent) re-computation. In this work we allow that node labels\footnote{Allowing relabeling is also called a labeling scheme.} may contain information that help with distance estimation and routing decisions, which gives rise to our second goal; keeping node labels small.\footnote{Usually the amount of information stored at nodes for routing and distance estimation is also considered. Since the nodes in our model are computationally unbounded we do not focus on that. Our lower bounds have also no restriction on the local information.}
The third goal is to speed up the actual packet forwarding process to minimize latency and alleviate congestion. Given a graph with edge weights corresponding to (e.g.) link-latencies, we want to minimize the largest detour any packet takes in relation to the corresponding shortest path. This is also known as \textit{stretch}. Analogously, for distance oracles we want to minimize the worst estimation error relative to the true distance.

\paragraph{Routing Schemes \& Distance Oracles in Distributed Networks} In this work we are interested solving the above problems in a distributed setting (c.f., Definition \ref{def:sync_msg_passing}). This has particular importance given the distributed nature of many real networks where routing problems are relevant (most prominently, the Internet) and where providing a centralized view of the whole network is prohibitively expensive. Note that we are interested in computing routing schemes and distance oracles for the {local communication network}, which is motivated by the fact that typically lots of packets are routed during an ongoing session and due bandwidth and cost constraints of the global mode. However, the global mode can be used to send the (relatively small) destination label to the source of a packet quickly, which can then be stored at that node for the duration of a session.


From an algorithmic standpoint, computing routing schemes and distance oracles is an inherently \textit{global} problem. That is, allowing only local communication (i.e., the \LOCAL model) it takes $\Omega(n)$ rounds to accomplish this (we provide a proof of this in Lemma \ref{lem:lower_bound_local}).\footnote{Any graph problem can be solved in $\bigO(n)$ rounds in \LOCAL by collecting the graph and solving the problem locally at some node. This makes global problems uninteresting for the \LOCAL model, unless communication restrictions are increased (c.f., \CONGEST) or decreased (c.f., \hybrid).} A similar observation can be made for the global communication mode (\NCC model). If we are only allowed to use global communication and each node initially only knows its incident edges in the local network, it takes $\tilOm(n)$\footnote{The $\tilO(\cdot)$ notation suppresses multiplicative terms that are polylogarithmic in $n$.} rounds to compute routing schemes and distance oracles (we show that in Lemma \ref{lem:lower_bound_ncc}).\footnote{Computing routing schemes in \NCC is somewhat artificial, as the need for routing schemes for a local network suggests that it exists and can be used.}
This article addresses the question whether the combination of the two communication modes in the \hybrid model can overcome the $\tilOm(n)$ lower bound of the individual modes \LOCAL and \NCC. 

Our answer to this is two-pronged. First we show that indeed, we can compute routing schemes and distance oracles significantly faster, for instance, we show $\tilO(n^{1/3})$ rounds and labels of size $\Theta(n^{2/3})$ suffice (c.f., Theorem \ref{thm:exact_algo}). Second, we show that the \hybrid model is not arbitrarily powerful by giving polynomial lower bounds for these problems (depending on the stretch) that hold even for relatively large labels and unbounded local memory. For instance, we show that it takes $\tilOm(n^{1/3})$ rounds to solve either problem exact, even for unweighted graphs and labels of size $\bigO(n^{2/3})$ (c.f., Theorem \ref{thm:lower_bound_unweighted}).
We provide numerous, more nuanced results, depending on stretch and the type of problem, summarized in the following.

\subsection{Contributions and Overview}

Our contributions and results are summarized in Table \ref{tbl:results}, which gives a simplified overview of our complexity results for the various forms of routing scheme and distance oracle problems.
Here we also want to give some intuition into how our techniques work and highlight how some of the results are generalized in the main part.

\begin{table}[ht]
	\begin{center}
	\renewcommand{\arraystretch}{1.2} 
	\begin{tabular}{@{}p{13mm}>{\centering\arraybackslash}cccc@{}}
		\toprule
		problem & stretch & complexity & label-size & reference\\
		\midrule
		\multirow{2}{*}{\shortstack[l]{distance\\ oracles}} & $3\!-\!\eps$ & $\tilOm(n^{1/3})$ & $\bigO(n^{2/3})^\dagger$ & Thm.\ \ref{thm:lower_bound_distance_oracle_small_stretch}\\
		& $\ell$ & $\tilOm\big(n^{1/f(\ell)}\big)^\ddagger$ & $\bigO\big(n^{2/f(\ell)}\big)^\ddagger$ & Thm.\ \ref{thm:lower_bound_distance_oracle_small_stretch}, \ref{thm:lower_bound_distance_oracle_large_stretch} \\
		\midrule
		\multirow{3}{*}{\shortstack[l]{\textit{stateless}\\ routing\\ schemes}} & $\sqrt{3}\!-\!\eps$ & $\tilOm(n^{1/3})$ & $\bigO(n^{2/3})^\dagger$ & Thm.\ \ref{thm:lower_bound_stateless_routing}\\
		& $\sqrt{5}\!-\!\eps$ & $\tilOm(n^{1/5})$ & $\bigO(n^{2/5})^\dagger$ & Thm.\ \ref{thm:lower_bound_stateless_routing}\\
		& $1\!+\!\sqrt{2}\!-\!\eps$ & $\tilOm(n^{1/7})$ & $\bigO(n^{2/7})^\dagger$ & Thm.\ \ref{thm:lower_bound_stateless_routing}\\
		\midrule
		\multirow{4}{*}{\shortstack[l]{\textit{stateful}\\ routing\\ schemes}} & $\sqrt{2}\!-\!\eps$ & $\tilOm(n^{1/3})$ & $\bigO(n^{2/3})^\dagger$ & Thm.\ \ref{thm:lower_bound_stateful_routing}\\
		& $\frac{5}{3}\!-\!\eps$ & $\tilOm(n^{1/5})$ & $\bigO(n^{2/5})^\dagger$ & Thm.\ \ref{thm:lower_bound_stateful_routing}\\
		& $\frac{7}{4}\!-\!\eps$ & $\tilOm(n^{1/7})$ & $\bigO(n^{2/7})^\dagger$ & Thm.\ \ref{thm:lower_bound_stateful_routing}\\
		& $\approx 1.78$ & $\tilOm(n^{1/11})$ & $\bigO(n^{2/11})$ & Thm.\ \ref{thm:lower_bound_stateful_routing}\\
		\midrule
		\multirow{2}{*}{\shortstack[l]{all on\\ \textit{unw.\ graphs}}} & exact & $\tilOm(n^{1/3})$ &  $\bigO(n^{2/3})^\dagger$ & Thm.\ \ref{thm:lower_bound_unweighted}\\
		& $1\!+\!\eps$ & $\tilO(n^{1/3}/\eps)$ & $\Theta(\log n)$ & Thm.\ \ref{thm:approx_algo}\\
		\midrule
		\multirow{2}{*}{\shortstack[l]{all on \textit{weigh.}\\ \textit{graphs}}} & exact & $\tilO(n^{1/3})$ & $\Theta(n^{2/3})$ & Thm.\ \ref{thm:exact_algo}\\
		& 3 & $\tilO(n^{1/3})$ & $\Theta(\log n)$ & Thm.\ \ref{thm:approx_algo}\\
		\bottomrule
	\end{tabular}
	\shrt{\begin{minipage}[c]{0.58\linewidth}}
	\begin{tablenotes}
		\footnotesize
		\item $\dagger$ The lower bound on round complexity holds any node labeling of at most that size.
		\item $\ddagger$ For some function $f(\ell)$ that is linear in $\ell$.
	\end{tablenotes}
	\shrt{\end{minipage}}
	\end{center}
	\caption{Selected and simplified contributions of this paper.}
	\label{tbl:results}
\end{table}

\paragraph{Lower Bounds Summary} Our main contribution revolves around computational lower bounds for computing distance oracles and stateless and stateful routing schemes in the \hybrid model. Lower bounds for approximations are summarized in the first three groups of Table \ref{tbl:results}. We also provide a lower bound on unweighted graphs, given in the first row of the fourth group of Table \ref{tbl:results}. Note that all lower bounds hold regardless of the allowed local memory. Moreover, our lower bounds hold for randomized algorithms with constant success probability. 

In the main part, our results are formulated for the more general \hybridpar{\infty}{\gamma} model, that is $\gamma$ appears as a parameter. For instance, our lower bound on unweighted graphs is in fact \smash{$\Omega(n^{1/3}/\gamma^{1/3})$} rounds for labels of size up to \smash{$c \cdot n^{2/3} \cdot \gamma^{1/3}$} for some $c > 0$ (c.f., Theorem \ref{thm:lower_bound_unweighted}), that is, we get a polynomial lower bound for the \hybridpar{\infty}{\gamma} model for any $\gamma \in \tilde o(n)$. For easier readability we plug in the ``standard'' \hybrid model with $\gamma = \tilO(1)$, which lets us hide $\gamma$ by using the $\tilOm$ notation.

\paragraph{Lower Bounds Overview} The general proof idea is based on information theory and plays out roughly as follows. We start out with a two party communication problem, where Alice is given the state of some random variable $X$ and needs to communicate it to Bob (c.f., Definition \ref{def:party_comm_problem}). Any communication protocol that achieves this needs to communicate $H(X)$ (Shannon entropy of $X$ \cite{Shannon1948}) bits in expectation (c.f., Corollary \ref{cor:lower_bound_two_party}), which is a consequence of the source coding theorem (replicated in Lemma \ref{lem:source_coding_theorem}).

In Section \ref{sec:node_communication_problem} we translate this to the \hybrid setting into what we call the \textit{node communication problem}. There, we have two sets of nodes $A$ and $B$, where nodes in $A$ ``collectively know'' the state of some random variable $X$ and need to communicate it to $B$ (for more precise information see Definition \ref{def:node_comm_problem}). We show a reduction (via a simulation argument) where a \hybrid algorithm that solves the node communication problem on sets $A$ and $B$ that are at sufficiently large distance in the local graph, can be used to derive a protocol for the two party communication problem (c.f., Lemma \ref{lem:reduce_comm_problems}). We conclude that for sets $A,B$ with distance at least $h$ it takes $\tilOm\big(\min(H(X)/n,h)\big)$ rounds to solve this problem (Theorem \ref{thm:lower_bound_node_communication}).

In Section \ref{sec:lower_bounds_unweighted} we give a reduction from the node communication problem to distance oracle and routing scheme computation. The goal is to encode some random variable $X$ with large entropy (super-linear in $n$) into some randomized part of our local communication graph such that some node set $A$ knows $X$ by vicinity. We construct such a graph $\Gamma$ (see Figure \ref{fig:lower_bound_basic}) from the complete bipartite graph $G_{k,k} = (A,E)$ and a (i.i.d.) random $k^2$-bit-string \smash{$X = (x_e)_{e \in E}$} with \smash{$H(X)=\Theta(k^2)$}. Then an edge $e \in E$ of $G_{k,k}$ is present in $\Gamma$ iff $x_e = 1$.

The nodes in $A$ collectively know $X$ since they are incident to the edges sampled from $G_{k,k}$. We designate $k$ nodes of $A$ (one side of the bipartition in $G_{k,k}$) as the ``target nodes''. Then we connect each target with a path of length $h$ to one of $k$ ``source nodes'' which will take the role of $B$ (see Figure \ref{fig:lower_bound_basic}). We show that if the nodes in $B$ learn the distances to the source nodes they also learn about the (non-)existence of the edges sampled from $G_{k,k}$ and can conclude the state of $X$ and thus have solved the node communication problem. Choosing the trade-off between $k$ and $h$ appropriately (roughly $n^{2/3}$ and $n^{1/3}$) we conclude \smash{$\tilOm\big(\min(H(X)/n,h)\big) = \tilOm(n^{1/3})$} rounds of communication must have taken place to solve the (exact) distance oracle problem (Theorem \ref{thm:lower_bound_unweighted}).

One caveat is that in the distance oracle problem the nodes are only supposed to give a distance to a target when also provided with the target-label. So we choose the labels sufficiently small such that the ``free information'', given in form of the labels of all targets, is negligible. We can allow labels of size $\bigO(n^{2/3})$ without changing the above narrative, see Theorem \ref{thm:lower_bound_unweighted}.
For routing schemes we have to adapt the graph $\Gamma$ a bit. We add a slightly longer alternative route from sources to targets (Figure \ref{fig:lower_bound_basic}, left side) and show that the existence of edges and thus the state of $X$ can be concluded from the first routing decision the sources have to make.

So far we got lower bounds only for exact solutions with the advantage that they hold on unweighted graphs. In Section \ref{sec:lower_bounds_approx} we show how to use graph weights to get lower bounds for approximation algorithms. For this we replace $G_{k,k}$ with a balanced, bipartite graph $G = (A,E)$ with $k$ nodes and girth $\ell$ (length of the shortest cycle in $G$). As before, the existence of an edge $e \in E$ in $\Gamma$ is determined by to a random bit string $X = (x_e)_{e \in E}$, c.f., Figure \ref{fig:lower_bound_enhanced}. If some edge $e \in E$ is not in $\Gamma$, then the detour in $\Gamma$ between the endpoints of $e$ is at least $\ell\!-\!1$ edges (otherwise $e$ closes a loop of less than $\ell$ edges). By assigning large weights to edges sampled from $G$, we can transform this into multiplicative detour of almost $\ell\!-\!1$. Similar to the idea in the unweighted case, any algorithm for distance oracles that has stretch \textit{slightly} smaller than $\ell\!-\!1$ can be used to solve the node communication problem, which takes \smash{$\tilOm\big(\min(H(X)/n,h)\big)$} rounds.

To optimize the lower bound we need to maximize the entropy $H(X)$ of \smash{$X = (x_e)_{e \in E}$}, i.e., the density of $G$.
However, it is well known that girth and density of a graph are opposing goals: a graph with girth $2g+1$ can have at most \smash{$\bigO\big(n^{1+1/g}\big)$} edges (c.f., \cite{Alon2001}, simplified in Lemma \ref{lem:girth_edges_bound}). This inherently limits the amount of information we can encode in $\Gamma$ and we show in Lemma \ref{lem:approx_lower_bound_groundwork} how graph density affects lower bounds for the node communication problem.
The good news is, that for some girth values, graphs that achieve their theoretical density limit actually exist and have been constructed (c.f., \citealp{Benson1966, Singleton1966}, simplified form given in Lemma \ref{lem:low_girth_graph_density}). For higher girth values, graphs that come close to that limit are known (c.f., \cite{Lazebnik1997}, simplified form in Lemma \ref{lem:high_girth_graph_density}).\footnote{There is a long standing conjecture that for each girth $2g\!+\!1$ there exists a graph that reaches the theoretical limit of \smash{$\bigO\big(n^{1+1/g}\big)$} edges (\cite{Erdoes1982}, c.f., Conjecture \ref{con:erdoes_girth}). Our tools can be used to generate new lower bounds in case new such graphs are found.}

Utilizing these graphs we achieve polynomial lower bounds for the \textit{distance oracle} problem for some small stretch values (c.f., Theorem \ref{thm:lower_bound_distance_oracle_small_stretch}) and for arbitrary constant stretch (c.f., Theorem \ref{thm:lower_bound_distance_oracle_large_stretch}). Theorem \ref{thm:lower_bound_distance_oracle_large_stretch} is heavily parametrized, but to sum it up in a simpler way: for any constant stretch $\ell$ we attain a polynomial lower bound of \smash{$\tilOm(n^{1/f(\ell)})$}, that is, $f(\ell)$ is constant as well (roughly \smash{$f(\ell) \approx \frac{3}{2}\ell$}).

For approximate \textit{routing schemes} we have to be more careful and also make a distinction between the stateless and stateful variant (c.f., Definitions \ref{def:stateless_routing}, \ref{def:stateful_routing}). 
The idea is the same as in the exact, unweighted case, however, since a wrong routing decision at the source can still be completed into a routing path of relatively good quality, the best stretch that can be achieved for lower bounds is limited (even more so for stateful routing, where a packet may ``backtrack''). In particular, allowing too much stretch can open up unwanted routing paths that mislead the sources in their conclusions about $X$. We forbid these unwanted routing paths using inequalities parametrized by the stretch and graph weights. Maximizing the stretch subject to these conditions we obtain the lower bounds for stretch values that are given in the second and third group of Table \ref{tbl:results} with details in Theorems \ref{thm:lower_bound_stateless_routing} and \ref{thm:lower_bound_stateful_routing}.

\paragraph{Upper Bounds Summary \& Overview} Our computational upper bounds (Given in Appendix \ref{sec:upper_bounds}) can be expressed more concisely due to the existence of efficient randomized algorithms for shortest path problems in the \hybrid model. In particular we draw on fast solutions for the so called \textit{random sources shortest paths problem} (RSSP) \cite{CensorHillel2021}, where all nodes must learn their distance to a set of i.i.d. randomly sampled nodes, say $S$. After solving RSSP, our strategy is to use the distance between a node $u$ and the nodes in $S$ as its label $\lambda(u)$. 

Roughly speaking, provided that $u$ is sufficiently ``far away'', a node $v$ can combine $\lambda(u)$ with its own distances to $S$ to compute its distance (estimate) to $u$. If $u$ is ``close'' then we can use the local network to compute the distance directly. While this gives us only distance oracles, it is relatively straight forward to also derive routing schemes. Simply speaking, we can always send a packet to a neighbor that has the best distance (estimate) to $u$ (some care must be taken for approximations). Note that this process is oblivious to previous routing decisions so the obtained routing scheme is stateless (c.f., Definition \ref{def:stateless_routing}).

A trade-off arises from the local exploration around nodes and the global computation depending on the size of $S$ (since we solve RSSP on $S$), which balances out to a round complexity of \smash{$\tilO(n^{1/3})$} with \smash{$|S| \in \tilO(n^{2/3})$} (similar trade-offs were observed for shortest paths problems in \cite{Augustine2020, Kuhn2020, CensorHillel2021}). For exact algorithms (distance oracles and routing schemes) we require labels of size \smash{$\Theta(n^{2/3})$} (however we can decrease the label size to \smash{$\Theta(n^{2/3-\zeta})$} at a cost of \smash{$\tilO(n^{1/3+\zeta})$} rounds, c.f., Theorem \ref{thm:exact_algo}). This is tight up to $\polylog n$ factors as is shown by the corresponding lower bound in Table \ref{tbl:results} group 4 line 1 (which holds even on unweighted graphs). 

For smaller labels we show that restricting $\lambda(u)$ to $u$'s closest node in $S$ gives good approximations. We obtain a 3-approximation on weighted graphs and a $(1\!+\!\eps)$ approximation on unweighted graphs in $\tilO(n^{1/3})$ rounds (assuming $\eps>0$ is constant) with labels of size $\bigO(\log n)$ (c.f., Theorem \ref{thm:approx_algo}). Compare this to our lower bounds: even much larger labels of size $\Theta(n^{2/3})$ do \textit{not} help to improve the runtime or the stretch by much, as this still takes $\tilOm(n^{1/3})$ rounds for stretch of $3\!-\!\eps$ for distance oracles on weighted graphs, and stretch $1$ on unweighted graphs (see Table \ref{tbl:results}).

\subsection{Related Work}

There was an early effort to approach hybrid networks from a theoretic angle \cite{Afek1990}, with a conceptually different model.\footnote{Essentially \cite{Afek1990} combines \LOCAL with a global channel where in each round \textit{one} node may broadcast a message, making it much weaker than the \hybrid model.} 
Research on the current take of the \hybrid model was initiated by \cite{Augustine2020} in the context of shortest paths problems, which most of the research has focused on so far. As shortest paths problems problems are closely related, we give a brief account of the recent developments.

\paragraph{Shortest Paths in the Hybrid Model} \cite{Augustine2020} introduced an information dissemination scheme to efficiently broadcast small messages to all nodes in the network. Using this protocol, they derive various solutions for shortest paths problems. For instance, for SSSP:\footnote{In the $k$ sources shortest paths problem ($k$-SSP) all nodes must learn their distance to $k$ dedicated source nodes. Then SSSP \smash{$\stackrel{\text{def}}{=}$} $1$-SSP, APSP \smash{$\stackrel{\text{def}}{=}$} $n$-SSP.} a $(1\!+\!\eps)$ stretch, \smash{$\tilO(n^{1/3})$}-round algorithm and a $(1/\eps)^{\bigO(1/\eps)}$-stretch, \smash{$\tilO(n^{\eps})$}-round algorithm.\footnote{\cite{Augustine2020} also gives an exact, \smash{$\tilO\big(\!\sqrt{\text{SPD}}\big)$}-round SSSP algorithm depending on the shortest path diameter SPD, using a completely different approach.}  Further, an approximation of APSP with stretch 3 in \smash{$\tilO(n^{1/2})$} rounds, which closely matches their corresponding \smash{$\tilOm(n^{1/2})$} lower bound (which holds for much larger stretch).
Subsequently, \cite{Kuhn2020} introduced a protocol for efficient routing of small messages between dedicated source-target pairs in the \hybrid model (not to be confused with routing schemes), which they use to solve APSP and SSSP {exactly} in \smash{$\tilO(n^{1/2})$} and \smash{$\tilO(n^{2/5})$} rounds, respectively. For computing the diameter they provide algorithms (e.g., a $3/2\!+\!\eps$ approximation in \smash{$\tilO(n^{1/3})$} rounds) and a \smash{$\tilOm(n^{1/3})$} lower bound.
\cite{CensorHillel2021} combines the techniques of \cite{Kuhn2020} with a densitiy sensitive approach, to solve $n^{1/3}$-SSP (thus SSSP) exactly and compute a $(1\!+\!\eps)$-approximation of the diameter in \smash{$\tilO(n^{1/3})$} rounds.\footnote{As we reuse some techniques of \cite{Augustine2020, Kuhn2020, CensorHillel2021}, we explain them in a bit more detail in Section \ref{sec:upper_bounds}).} \cite{CensorHillel2021a} uses density awareness in a different way to improve SSSP to \smash{$\tilO(n^{5/17})$} rounds for a small stretch of $(1\!+\!\eps)$. 
\cite{Anagnostides2021} derandomized the dissemination protocol of \cite{Augustine2020} to obtain a deterministic APSP-algorithm with stretch \smash{$\frac{\log n}{\log\log n}$} in \smash{$\tilO(n^{1/2})$} rounds.
For classes of sparse graphs (e.g., cactus graphs) \cite{Feldmann2020} demonstrates that $\polylog n$ solutions are possible even in the harsher hybrid combination \CONGEST and \NCC.

\paragraph{Routing Schemes in Distributed Models} Routing schemes in a hybrid network model have been pioneered in \cite{Coy2021}. Their work can be contrasted to this article in three main ways: First, they consider specific types of local graphs, namely ``hole-free'' grid graphs and unit disc graphs (UDGs) (which have practical relevance in the context of wireless local networks), whereas this work focuses on general graphs. Second, \cite{Coy2021} provides upper bounds based on a routing scheme on a grid-graph abstraction of UDGs, whereas the main technical contribution of this work are lower bounds. Third, we consider hybrid networks with unlimited \LOCAL edges (which strengthens our lower bounds), whereas \cite{Coy2021} gives a $\bigO(\log n)$ round algorithm with labels and local tables of size $\bigO(\log n)$ and stretch 1 and constant stretch on certain grid-graphs and UDGs, respectively, which holds even for the stricter combination of \CONGEST and \NCC. Note that our lower bounds supplement the work of \cite{Coy2021} in that it shows the necessity to narrow the scope (for instance to particular graph classes) in order to achieve $\bigO(\log n)$ round algorithms with small labels and stretch.

A related line of work investigates the round complexity of implementing routing schemes and distance oracles in the \CONGEST model, where the challenge is that only small local messages can be used. Here, a \smash{$\tilOm(n^{1/2} \!+\! D_G)$} lower bound for computing routing schemes is implied by \cite{Sarma2012} (even for polynomial labels, whereas $D_G$ is the diameter of $G$). For distance oracles in \CONGEST, \cite{Izumi2014} gives a lower bound of \smash{$\Omega(n^{\frac{1}{2}+\frac{1}{5k}})$} with stretch $2k$ for graphs with small diameter assuming Erd\H{o}s' girth conjecture (c.f., Conjecture \ref{con:erdoes_girth}). Their bound holds for small labels without this assumption.
A line of papers narrows the gap to these lower bounds \cite{Lenzen2013a, Lenzen2015, Elkin2016}. For instance \cite{Elkin2016} achieves a routing scheme with stretch $\bigO(k)$, routing tables of size $\tilO(n^{1/k})$, labels of size $\tilO(k)$ in $\bigO\big(n^{1/2 + 1/k} \!+\! D_G\big) \cdot n^{o(1)}$ rounds.

\paragraph{Routing Schemes as Distributed Data Structure} In this branch of research the goal is often to optimize the trade-off between stretch and the size of the local routing tables. It is well established that $\tilO(n)$ bit of memory per node always suffices for routing schemes\footnote{As $\tilO(n)$ bit of memory allows each node to store the next node for every destination.} \cite{Peleg1989}. For routing schemes on general graphs with small stretch ($<\!3$) $\Omega(n)$ bit are required \cite{Peleg1989, Thorup2001}.\footnote{With $o(n)$ local memory a $\Omega(n)$-degree node has to ``forget'' neighbors, necessarily generating a detour when routing to the neighbors that have been forgotten. This lower bound is part of the reason why we do not analyze local memory in this work.} More generally, for a stretch smaller than $2k\!+\!1$ it is known that local tables of size $\Omega({n^{1/k}})$ bit are required for small values $k \in \{1,2,3,5\}$, see \cite{Thorup2001}. This is true for all parameters of $k$ if one believes Erd\H{o}s' girth conjecture (see Conjecture \ref{con:erdoes_girth}) \cite{Thorup2001}. The reliance on this conjecture can be dropped for weaker \textit{name-independent} routing schemes, where nodes can \textit{not} be re-labeled and \cite{Abraham2006} shows that this requires \smash{$\Omega({(n\log n)^{1/k}})$} memory for any stretch smaller than $2k\!+\!1$. 
On the positive side, \cite{Thorup2001} shows that a stretch of $\bigO(k)$ with \smash{$\tilO({kn^{1/k}})$} memory can be achieved. The concrete stretch is $4k\!-\!5$ or even $2k\!-\!1$ (i.e., almost optimal) if \textit{handshaking} is allowed (an initial exchange of a message between source and destination before a packet is routed). The lower bound for name independent routing has been contrasted with a scheme that achieves $O(k)$ stretch with memory size \smash{$\tilO(k^2n^{1/k})$} \cite{Abraham2006a}.

\subsection{Preliminaries}
\label{sec:preliminaries}

\paragraph{General Definitions} The scope of this paper is solving graph problems, typically in the undirected communication graph. Let $G = (V,E)$ be undirected. Edges have {weights} $w: E \to [W]$, where $W$ is at most polynomial in $n$, thus the weight of an edge and of a simple path fits into a $\bigO(\log n)$ bit message.\footnote{In this article, $\log$ functions are always to the base of 2.} A graph is considered {unweighted} if $W=1$. Let $w(P) = \sum_{e \in P}w(e)$ denote the length of a path $P \subseteq E$. Then the \emph{distance} between two nodes $u,v \in V$ is
\lng{\[
d_G(u,v) := \!\min_{\text{$u$-$v$-path } P} w(P).
\]}
\shrt{$d_G(u,v) := \!\min_{\text{$u$-$v$-path } P} w(P).$}
A path with smallest length between two nodes is called a \emph{shortest path}.
Let $|P|$ be the number of edges (or \emph{hops}) of a path $P$.
The \emph{hop-distance} between two nodes $u$ and $v$ is defined as: 
\lng{\[
\hop_G(u,v) := \!\min_{\text{$u$-$v$-path } P} |P|.
\vspace*{-1mm}
\] }
\shrt{$\hop_G(u,v) := \!\min_{\text{$u$-$v$-path } P} |P|.$}
We generalize this for sets $U,W \subseteq V$ (whereas $\hop_G(v,v) := 0$):
\lng{\[
\hop_G(U,W) := \!\min_{u \in U, w \in W} \hop_G(u,w).
\vspace*{-1mm}
\] }
\shrt{$\hop_G(U,W) := \!\min_{u \in U, w \in W} \hop_G(u,w).$}
The \emph{diameter} of $G$ is defined as: 
\lng{\[
D_G:= \max_{u,v \in V} \hop_G(u,v).
\]}
\shrt{$D_G:= \max_{u,v \in V} \hop_G(u,v).$}
Let the \emph{$h$-hop distance} from $u$ to $v$ be: 
\lng{\[
d_{G,h}(u,v) := \!\!\min_{{\text{$u$-$v$-path } P, |P| \leq h }}\, w(P).
\]}
\shrt{$d_{G,h}(u,v) := \!\!\min_{{\text{$u$-$v$-path } P, |P| \leq h }}\, w(P).$}
If there is no $u$-$v$ path $P$ with $|P|\leq h$ we define $d_{h}(u,v) := \infty$. We drop the subscript $G$, when $G$ is clear from the context. In this paper we consider the following problem types:

\begin{definition}[Distance Oracles]
	\label{def:distance_oracles}
	Every node $v \in V$ of a graph $G= (V,E)$ needs to compute a \textit{label} $\lambda(v)$ and an \textit{oracle function} $o_v : \lambda(V) \to \mathbb{N}$, such that $o_v(\lambda(u)) \geq d(u,v)$ for all $u \in V$.
	An oracle function $o_v$ is an $(\alpha, \beta)$-approximation if $o_v(\lambda(u)) \leq \alpha \cdot d(u,v) + \beta$ for all $u,v \in V$, that is, $\alpha, \beta$ are the \textit{multiplicative} and \textit{additive} approximation error, respectively. We speak of a \textit{stretch} of $\alpha$ in case of an $(\alpha, 0)$-approximation. If the stretch is one, we call $o_v$ \textit{exact}.
\end{definition}

\begin{definition}[Stateless Routing Scheme] 
	\label{def:stateless_routing}
	Every node $v \in V$ of a graph $G= (V,E)$ needs to learn a \textit{label} $\lambda(v)$ and a \textit{routing function} (sometimes called ``table'') $\rho_v : \lambda(V) \to N(v)\cup \{v\}$ where $N(v)$ are adjacent nodes of $v$ in $G$ (whereas we formally set $\rho_v(\lambda(v)) := v$). The functions $\rho_v$ must fulfill the following correctness condition. 
	Let $v_0 := v$ and recursively define \smash{$v_i := \rho_{v_{i-1}}(\lambda(u))$}. 
	Then the routing functions $\rho_v, v \in V$ must satisfy $v_h = u$ for some $h \in \mathbb{N}$.
	Let $P_\rho(u,v)$ be the path induced by the visited nodes $v_0, \dots, v_h$. We call $\rho$ an $(\alpha, \beta)$-approximation if $w(P_\rho(u,v)) \leq \alpha d(u,v) + \beta$ for all $u,v \in V$.
\end{definition}

\begin{definition}[Stateful Routing Scheme] 
	\label{def:stateful_routing}
	This is mostly defined as in the stateless case, with the difference that the routing function $\rho_v$ can additionally depend on the information gathered along the path that has already been visited by a packet (which would be stored in its header). Note that in this means that the routing path defined by such a function $\rho$ is not necessarily simple (i.e., might have loops).
\end{definition}

\begin{definition}[Randomized Graph Algorithms] 
	\label{def:rand_algo}
	We say that an algorithm has success probability $p$, if it succeeds with probability at least $p$ on every possible input graph (however, some of our results are restricted to unweighted graphs as input). Specifically, for our upper bounds we aim for success \textit{with high probability} (w.h.p.), which means with success probability at least $1-\frac{1}{n^c}$ for any constant $c>0$.
\end{definition}



\lng{\section{Upper Bounds} 
\label{sec:upper_bounds}

\lng{The first part of this paper is to} \shrt{In this part of the paper we} derive algorithms that compute routing schemes and distance oracles in the \hybrid model\lng{, which we consider as a warm-up and complementary to the subsequent section on lower bounds}.
We can draw on the techniques and fast algorithms for shortest paths problems from \cite{Augustine2020, Kuhn2020, CensorHillel2021} (where most of the heavy lifting occurs) and show how to leverage these to obtain distance oracles and routing schemes efficiently in the \hybrid model. We start with a quick introduction to the techniques we use.

{\subsection{Techniques}
	
	Skeleton graphs were first used by \cite{Ullman1991} and became one of the main tools used in the context of shortest path algorithms in the \hybrid model (c.f., \cite{Augustine2020,Kuhn2020,CensorHillel2021}). Simply speaking, a skeleton is a minor of a graph $G$ consisting of a (usually relatively small) set of nodes sampled with some probability $\frac{1}{x}$ and virtual edges formed between sampled nodes at most $\tilO(x)$ hops apart with weights such that distances in the skeleton graph correspond to those in $G$ w.h.p. 
	
	The usual approach is to solve a given shortest path problem on an appropriately sized skeleton graph by leveraging the global communication provided by the \hybrid model and then extend that solution to the whole graph using local communication. While at this point it is not necessary to fully characterize skeleton graphs anymore (we use existing algorithms out of the box), we still have to use the subsequent property for a random sampling of nodes (from which all required properties of a skeleton graph are derived).
	
	Assume each node joins some set $S \subseteq V$ independently and identically distributed (i.i.d.) with probability $1/x$ for some $x > 1$. The expected number of nodes joining $S$ is $\mathbb E(|S|) = n/x$, and we also have $|S| \in \Theta(n/x)$ w.h.p.\ if $\mathbb E(|S|)$ is sufficiently large (which a simple application of the Chernoff bound given in Lemma \ref{lem:chernoffbound} shows).
	Furthermore, there exists some $h \in \tilO(x)$ such that for any $u,v \in V$ there will be a sampled node on some shortest $u$-$v$-path $P$ at least every $h$ hops for any $u,v \in V$ w.h.p. This is formalized as follows (the proof can be found at the end of Appendix \ref{apx:generalnotations}). 
	
	\begin{lemma}[c.f., \cite{Kuhn2020}, \cite{Augustine2020}]
		\label{lem:node_sampling}
		Let $G=(V,E)$. Let each node join some set $S\subseteq V$ i.i.d.\ with probability $\frac{1}{x}$. Then there is a constant $\xi \!>\! 0$, such that for any $u,v \!\in\! V$ with $hop(u,v) \!\geq\! h := \xi x \ln n$, there is at least one shortest path $P$ from $u$ to $v$, such that any sub-path $Q$ of $P$ with at least $h$ nodes contains a node in $S$ w.h.p.
	\end{lemma}
	
	We are interested in the distances to the sampled nodes, which can be formalize as follows.
	
	\begin{definition}[Random Sources Shortest Paths (RSSP)]	
		Given a subset of nodes (sources) that were sampled i.i.d.\ with probability $1/x$ for some $x \geq 1$ from an undirected, weighted graph $G$. The RSSP problem is solved when every node in the network has learned its distance to each of the random source nodes. 
	\end{definition}
	
	We provide a rough overview how RSSP was solved (simplified for \smash{$x = n^{1/3}$}) for readers unfamiliar with the topic. First, a skeleton graph of $G$ is constructed on the \smash{$\Theta(n^{2/3})$} sampled nodes as in \cite{Augustine2020}. Note that edges of that skeleton are at most \smash{$\tilT(n^{1/3})$} hops apart, which means that one round of the \LOCAL model can be simulated on that skeleton graph in \smash{$\tilT(n^{1/3})$} real rounds. 
	Second, \cite{Kuhn2020} showed that a round of the \CC model\footnote{In the \CC model, in each round each node is allowed to send a (different) $\bigO(\log n)$ message to every node.} can also be simulated in \smash{$\tilT(n^{1/3})$} rounds on the skeleton graph using a routing protocol tailored to the \hybrid model to efficiently communicate $\bigO(\log n)$ bit messages between pairs of senders and receivers.
	
	Third, \cite{CensorHillel2021} observed that skeleton nodes with a high degree in the skeleton graph can learn a lot of information by relying on the bandwidth of their neighbors in the skeleton. With this observation they use \smash{$\tilO(1)$} simulated \LOCAL and \CC rounds (of \smash{$\tilT(n^{1/3})$} real rounds each) to compute ``tiered oracles'' (not to be confused with distance oracles in this paper) meaning that each node with a degree in a certain ``tier'' (exponential degree class \smash{$\{2^{i-1}\!\!\!, \dots, 2^{i}\!-\!1\}$ with $i \in [\lceil\log n\rceil]$}) learns the subgraph of the skeleton graph induced by nodes of degree in its own tier or below.
	
	Iteratively, one can solve the all pairs shortest paths problem on the skeleton. Initially, the nodes of the highest tier know the whole skeleton graph and can send their distance to every node in one \CC round. This information enables the nodes in the next lower tier to compute their distance to every node in the skeleton. Then the process repeats (for each of the $\tilO(1)$ tiers) until eventually all pairs shortest paths is solved on the skeleton. All other nodes in $G$ can learn their distance to each skeleton node from those within $\tilO(n^{1/3})$ hops. The following lemma summarizes this.
	
	\begin{lemma}[c.f., \cite{CensorHillel2021}\footnote{For our applications, it is convenient to rephrase the result of \cite{CensorHillel2021}, which considers sampling probabilities $n^{y-1}$ for $0<y<1$ with running time \smash{$\tilO(n^{1/3}+n^{2y-1})$}. This corresponds to the variant given here via the substitution \smash{$y = 1-\frac{\log x}{\log n}$}.}]
		\label{lem:rssp}
		There is an algorithm that solves the random sources shortest path problem for sampling probability $1/x$ for some $x \geq 1$ exactly and w.h.p.\ in \smash{$\tilO(n^{1/3}+n/x^2)$} rounds.
	\end{lemma}
	
	This result is almost tight (up to $\polylog(n)$ factors) for $x = n^{1/3}$ (i.e., there are $|S| = \Theta(n^{2/3})$ sources) due to a corresponding lower bound by \cite{CensorHillel2021} which was slightly adapted for random sources from the $\tilOm(\!\sqrt{k})$ lower bound for the $k$-sources shortest path problem given in \cite{Kuhn2020}.}

\subsection{Base Algorithm}

{We combine the previous two techniques to compute exact distance oracles and routing schemes in the \hybrid model.} We will presume a subroutine \code{explore($h$)} that floods all graph information for $h$ hops such that afterwards each node $u \in V$ knows $\calB_{u,h}$ defined as the subgraph induced by the nodes within $h$ hops around $u$. Furthermore we presume a subroutine \code{solve-rssp} that solves RSSP on the sampled nodes in accordance with Lemma \ref{lem:rssp}. 

Algorithm \ref{alg:base_algo} gives an overview. The parameters $x$ and $h$ are tuning parameters which will be specified for the task at hand. Algorithm \ref{alg:base_algo} consists of two parts: first all nodes collect the information required to compute distance oracles and routing schemes. Second, in the last three lines, the computation of the label $\lambda(v)$, oracle function $o_v$ and routing function $\rho_v$ takes place (c.f., Definitions \ref{def:distance_oracles}, \ref{def:stateful_routing}). We will override these sub-procedures to adapt the algorithm for exact and approximate results and for the special case of unweighted graphs. The steps are described on a node-level.

\begin{algorithm}[H]
	\caption{\code{base-algorithm($x,h$)}\Comment{\textit{$h$ not smaller as in Lem.\ \ref{lem:node_sampling}}}}
	\label{alg:base_algo}
	\begin{algorithmic}
		\State $v$ joins $S$ with probability $\frac{1}{x}$
		\State \code{explore($h$)} \Comment{\textit{$v$ learns $\calB_{v,h}$}}	
		\State \code{solve-rssp} \Comment{\textit{with source nodes $S$}}
		\State \code{create-label} \Comment{\textit{computation of $v$'s problem-specific label}}
		\State \code{create-oracle-function} \Comment{\textit{output  $o_v$ as in Def.\ \ref{def:distance_oracles}}}
		\State \code{create-routing-function} \Comment{\textit{output $\rho_v$ as in Def.\ \ref{def:stateless_routing}}}
	\end{algorithmic}
\end{algorithm}

\begin{lemma}
	\label{lem:base_algo_runtime}
	Given that the last three lines have constant running time, Algorithm \ref{alg:base_algo} takes $\tilO(n^{1/3}+n/x^2 + h)$ rounds.
\end{lemma}

\begin{proof}
	The time consuming steps are \code{explore($h$)}, i.e., the flooding of the graph for $h \in \tilOm(x)$ rounds and solving the RSSP with $\tilO(n^{1/3}+n/x^2)$ (c.f., Lemma \ref{lem:rssp}).
\end{proof}

\subsection{Exact Distance Oracles \& Routing Schemes}

First, \code{create-label} computes $v$'s label $\lambda(v):=\big\{\big(\text{ID}(s),d(s,v)\big)\mid s \in S\big\}$ corresponding to the identifiers of all nodes in $S$ and the associated distances that are known from solving the RSSP problem. 
Second, for a given label $\lambda(u)$ of some destination $u \in V$ \code{create-oracle-function} outputs the following function
\begin{equation}
\label{eq:exact_oracle_function}
o_v(\lambda(u)) = \min\Big(d_h(v,u), \,\min_{s\in S} d(v,s) + d(s,u)\Big)
\end{equation}
Note that in the above equation we set $d(v,v)=0$ and $d_h(v,u) = \infty$ for $u \notin \calB_{v,h}$. The $h$-hop distances $d_h(v,u)$ are known since $v$ knowns $\calB_{v,h}$, the distances $d(v,s), s \in S$ are known from solving the RSSP problem on $S$ and the distances between pairs $d(s,u), u \in S$ are part of the label $\lambda(u)$, thus $o_v(\lambda(u))$ can be computed by $v$.

Exchanging the computed distance oracles between neighbors  problem gives us sufficient information to solve the routing problem. Roughly speaking, the next node on a shortest path is given by a neighbor that minimizes the distance to the destination of the packet (a little bit of care has to be taken for proving this, though).

More precisely, the subroutine \code{create-routing-function} will first do a single round of communication so that each node $v$ learns the oracle function $o_w$ of each of its neighbors $w$ in the local network. Then we use this knowledge to pass the packet only to a neighbor whose distance oracle to the target decreases by the weight of that edge (pick an arbitrary one if there is more than one).
\begin{align}
\label{eq:exact_routing_function}
& \rho_v(\lambda(u)) \in \big\{z \in N(v) \mid o_z(\lambda(u)) = o_v(\lambda(u)) - w(v,z) \big\}
\end{align}
Note that additional knowledge of $o_z(\lambda(u))$ for $z \in N(v)$ is sufficient to compute $\rho_v(\lambda(u))$.
It remains to analyze the resulting algorithm and give the remaining correctness arguments, which we do in the proof of the following theorem. Note that here we obtain a trade off between label size and the running time, which we formalize with a parameter $\zeta$.

\begin{theorem}
	\label{thm:exact_algo}
	For any $\zeta \geq 0$ \emph{exact} distance oracles and \emph{stateless} routing schemes with labels of size \smash{$\bigO(n^{\frac{2}{3}-\zeta})$} bits can be computed in \smash{$\tilO(n^{\frac{1}{3}+\zeta})$} rounds in the \hybrid model w.h.p.
\end{theorem}

\begin{proof}
	Note that for $\zeta \geq 2/3$ the problem becomes trivial (any graph problem can be solved in $n$ rounds in \hybrid), so we assume $\zeta < 2/3$. We choose \smash{$x = n^{\frac{1}{3}+\zeta}$} and $h = \tilT(x)$ as in Lemma \ref{lem:node_sampling}, then the runtime follows from Lemma \ref{lem:base_algo_runtime}. The number of sampled nodes is \smash{$|S| \in \Theta(n^{\frac{2}{3}-\zeta})$} (a simple application of Lemma \ref{lem:chernoffbound}). As the label $\lambda(v)$ contains information of size $\bigO(\log n)$ bits for each node $s \in S$ (recall that distances are polynomial in $n$), $\lambda(v)$ requires \smash{$\bigO(n^{\frac{2}{3}-\zeta}\log n)$} bits. We can shift the $\log n$ factor from the label size into the runtime (where it is absorbed by the $\tilO$ notation) with a substitution \smash{$\zeta= \zeta' + \frac{\log\log n}{\log n}$}.
	
	We already established in the algorithm description that the required information to compute $\lambda(v)$, $o_v$ and $\rho_v$ is present at the node $v$ from executing \code{explore($h$)} and \code{solve-rssp}. 		
	It remains to show that $o_v(\lambda(u)) = d(v,u)$. If there is a shortest $v$-$u$-path with at most $h$ hops, then the first argument of the outer $\min$ function in Equation \eqref{eq:exact_oracle_function} corresponds to $d(v,u)$. If all shortest $v$-$u$-paths have more than $h$ hops then there must be a node $s \in S$ on one such path by Lemma \ref{lem:node_sampling}, which means that the second argument of the outer $\min$ function in Equation \eqref{eq:exact_oracle_function} corresponds to $d(v,u)$.
	
	Regarding the routing function, we observe that $z \in N(v)$ is the next node on \emph{some} shortest path from $v$ to $u$ if and only if $d(v,u) = w(v,z) + d(z,u)$, which is equivalent to $o_z(u) = o_v(u) - w(v,z)$ as $o_v,o_z$ are exact distance oracles. Since $G$ is connected, there is always a neighbor of $v$ on some shortest path to $u$ so the set in Equation \eqref{eq:exact_routing_function} is not empty.	
	By always choosing the next node on \emph{some} shortest path from the current node to $u$ (c.f., Equation \eqref{eq:exact_routing_function}), we will naturally obtain a shortest $v$-$u$-path overall. Note that the obtained routing scheme is oblivious to previous routing decisions.
\end{proof}

\subsection{Approximate Solutions}
\label{sec:upper_bounds_approx}

Note that if we do not want to compromise on the running time the label size in Theorem \ref{thm:exact_algo} is $\tilOm(n^{2/3})$. However, we can compromise on the stretch in order to significantly decrease the label size as we show in the following.

Subroutine \code{create-label} of $v$ computes the much smaller label $\lambda(v):=\big(\text{ID}(s_v),d(s_v,v)\big)$, where $s_v := \argmin_{s \in S} d(s,v)$ corresponds to the identifier and distance of the \textit{closest} node in $S$.

Given some label $\lambda(u):=\big(\text{ID}(s_u) ,d(s_u,v)\big)$ then node $v$ outputs the following function in subroutine \code{create-oracle-function}, which is the shorter of the $h$-hop distance to $u$ and the shortest $v$-$u$-path that transits $s_u$.
\begin{equation}
\label{eq:approx_oracle_function}
o_v(\lambda(u)) = \min\big(d_{h}(v,u), \,d(v,s_u) + d(s_u,u)\big).
\end{equation}

We start by showing that the oracle function $o_v$ from above gives a decent approximation (with proofs similar to those for approximations for the single source shortest paths problem in \cite{Augustine2020}).

\begin{lemma}
	\label{lem:algo-approx-weighted}
	For $u,v \in V$ it is $o_v(u) \leq 3d(v,u)$ w.h.p. (for $o_v$ from Eq.\ \eqref{eq:approx_oracle_function}).
\end{lemma}

\begin{proof}
	If there is a shortest $v$-$u$-path with at most $h$ hops, then we actually have $o_v(\lambda(u)) = d_{h}(v,u) = d(v,u)$ (c.f., first argument of Equation \eqref{eq:approx_oracle_function}) analogous to the proof of Theorem \ref{thm:exact_algo}. 	
	Else, all shortest $v$-$u$-paths have more than $h$ hops. Then by Lemma \ref{lem:node_sampling} w.h.p.\ there must be a node $\hat s \in \calB_{u,h}$ such that 
	\[
		d(v,u) = d(v,\hat s) + d(\hat s,u).\tag*{(i)}
	\] 
	Let $s_u := \argmin_{s \in S} d(s,v)$ from the label $\lambda(u)$. Then
	\[
	d(s_u,u) \stackrel{\text{Def.\ of } s_u}{\leq} d(\hat s,u) \stackrel{\text{(i)}}{=} d(v,u) - d(v,\hat s) \leq d(v,u).\tag*{(ii)}
	\]
	Then, for the second argument of Equation \eqref{eq:approx_oracle_function} we get
	\begin{align*}
	o_v(\lambda(u)) & \leq d(v,s_u) + d(s_u,u)\tag*{\textit{c.f., Equation \eqref{eq:approx_oracle_function}}}\\
	& \leq  d(v,u) + d(u,s_u) + d(s_u,u) \tag*{\textit{triangle ineq.}}\\
	& = d(v,u) + 2 d(u,s_u) \leq 3d(v,u).\tag*{\textit{by ineq.} {(ii)}}
	\end{align*}
\end{proof}

We can also express the approximation as additive error depending on largest weight $W'$ on some $v$-$u$-path and the hop distance $h'$ between sampled nodes on that shortest path.

\begin{lemma}
	\label{lem:algo-approx-unweighted}
	For $u,v \in V$ let $W'$ be the largest weight on \emph{some} shortest $v$-$u$-path. Then $o_v(\lambda(u)) \leq d(v,u) +2h' W'$ w.h.p.\ (for $o_v$ from Eq.\ \eqref{eq:approx_oracle_function}) where $h'\in \tilO(x)$ denotes the value from Lemma \ref{lem:node_sampling}.
\end{lemma}

\begin{proof}
	Fix a shortest $v$-$u$-path $P$. Let $W'$ be the largest weight on $P$. By Lemma \ref{lem:node_sampling} there is a sampled node $\hat s \in S$ on $P$ with $\hop(\hat s,u) \leq h'$ w.h.p.\footnote{Note for the union bound to work, Lemma \ref{lem:node_sampling} is phrased such that for each pair $u,v \in V$ there is just \textit{one} shortest $v$-$u$-path (out of possibly exponentially many) which has a sampled node at last every $h'$ hops. In fact in the proof we can fix one specific such path for each pair ourselves and the union bound still works.} This implies $d(\hat s,u) \leq h'W'$. Let $s_u := \argmin_{s' \in S}  d(s',v)$ (from the label $\lambda(u)$). By definition of $s$ we have $d(s_u,u) \leq d(\hat s,u) \leq h'W'$. Then we get
	\begin{align*}
	o_v(\lambda(u)) & \leq d(v,s_u) + d(s_u,u)\tag*{\textit{c.f., Equation \eqref{eq:approx_oracle_function}}}\\
	& \leq  d(v,u) + d(u,s_u) + d(s_u,u) \tag*{\textit{triangle ineq.}}\\
	& \leq d(v,u) + 2 h'W'.\tag*{\qedhere}
	\end{align*}
\end{proof}

%
%

Similar to Equation \eqref{eq:exact_routing_function}, we define $\rho_v(\lambda(u))$ as a node $z$ that is adjacent to $v$ such that the oracle function $o_z(\lambda(u))$ decreases by \textit{at least} the weight of the corresponding edge, i.e., $o_z(\lambda(u)) \leq o_v(\lambda(u)) - w(v,z)$ (we will show that there is always a neighbor of $v$ that fulfills these criteria and if several do, we choose one arbitrarily). Then subroutine \code{create-routing-function} outputs
\begin{align}
\label{eq:approx_routing_function}
& \rho_v(\lambda(u)) \in \big\{z \in N(v) \mid o_z(\lambda(u)) \leq o_v(\lambda(u)) - w(v,z) \big\}
\end{align}



\begin{lemma}
	\label{lem:stretch-routing}
	For any $u,v \in V$ the routing functions in Equation \eqref{eq:approx_routing_function} induce a $v$-$u$-path $P$ with $w(P) \leq o_v(\lambda(u))$ w.h.p.
\end{lemma}

\begin{proof}
	Assume a packet is forwarded along a path $P$ containing the nodes $v =: v_1, \dots , v_k := u$ with $o_{{i+1}} \leq o_{i} - w(v_i,v_{i+1})$ (whereas we abbreviate $o_i := o_{v_{i}}(\lambda(u))$ and we have $o_k = o_{u}(\lambda(u)) = 0$) then we can upper bound $w(P)$ with a telescoping sum
	\[
		w(P) = \! \sum_{i=1}^{k-1} \! w(v_i,v_{i+1})
		\leq \! \sum_{i=1}^{k-1} \! \big( o_i - o_{i+1} \big)
		= o_1 - \underbrace{\smash{o_k}}_{=0} = o_{v}(\lambda(u)).
	\]
	It remains to prove the existence of $P$. For that we show that each node $v$ has a neighbor $z$ that fulfills the requirement $o_z(\lambda(u)) \leq o_v(\lambda(u)) - w(v,z)$. We have to unwrap the definition of $o_v(\lambda(u))$ to make this argument.
	
	\textit{Case 1:} if $o_v(\lambda(u)) = d_{h}(v,u)$ (first argument of the $\min$ function in Equation \eqref{eq:approx_oracle_function}), then this means there is a $h$-hop path $Q$ from $v$ to $u$ that fulfills $w(Q) = o_v(\lambda(u))$. 
	Clearly, the next $z \in N(v)$ on $Q$ is one hop closer to $u$ and also sees the sub-path $Q_{z,u} \subseteq Q$ within its ball $\calB_{z,h}$ so by Equation \eqref{eq:approx_oracle_function} $z$'s distance estimate must be at least as good: $o_z(u) \leq w(Q_{z,u}) = o_v(\lambda(u)) - w(v,z)$.
	
	\textit{Case 2:} if $o_v(\lambda(u)) = d(v,s_u) + d(s_u,u)$, then this means our distance estimation via $s_u$ is at least as good as the $h$-hop distance to $u$ (c.f., second argument of $\min$ in Equation \eqref{eq:approx_oracle_function}). Let $Q$ be a path concatenated from two shortest paths $Q_{v,s_u} \cup Q_{s_u,u}$ from $v$ to $s_u$ and from $s_u$ to $u$, respectively. 
	
	Note that $\hop(s_u,u) \leq h$, as otherwise there would be a node $s_u' \in S$, $s_u'\neq s_u$ on $Q_{s_u,u}$ by Lemma \ref{lem:node_sampling}, which would be closer to $u$ thus contradicting the definition $s_u := \argmin_{s \in S} d(s,u)$ in the label $\lambda(u)$. This means that we can assume $v \neq s_u$ since otherwise we can apply case 1.	
	Let $z \in N(v)$ be the next node on $Q_{v,s_u}$ (which contains at least the 2 nodes $v$ and $s_u$, but $z = s_u$ is possible). Then
	\begin{align*}
		& d(v,s_u) + d(s_u,u) = w(v,z) + d(z,s_u) + d(s_u,u)\\
		\Longrightarrow \quad & o_v(\lambda(u)) - w(v,z) = d(z,s_u) + d(s_u,u) \geq o_z(\lambda(u)).\hfill\qedhere
	\end{align*}
\end{proof}


The following Theorem gives a summary of the obtained approximation algorithms for which we utilize the previous lemmas.

\begin{theorem}
	\label{thm:approx_algo}
	Distance oracles and \emph{stateless} routing schemes with label-size $\bigO(\log n)$ can be computed in \hybrid w.h.p.\ and
	\begin{itemize}
		\item stretch 3 in \smash{$\tilO(n^{1/3})$} rounds on weighted graphs,
		\item stretch $1\!+\!\eps$ for $0 \!<\! \eps \!\leq\! 1$ in \smash{$\tilO\big(\frac{n^{1/3}}{\eps}\big)$} rounds on unweighted graphs.
	\end{itemize}
\end{theorem}

\begin{proof}
	The label $\lambda(v)$ contains only the distance to the closest sampled node in $S$ which is $\bigO(\log n)$ bits. For the first result on weighted graphs we invoke Algorithm \ref{alg:base_algo} with $x = n^{1/3}$ and $h \in \tilT(x)$ as in Lemma \ref{lem:node_sampling}, then the runtime follows from Lemma \ref{lem:base_algo_runtime} and the stretch for the resulting distance oracles follows from Lemma \ref{lem:algo-approx-weighted}. Note that by Lemma \ref{lem:stretch-routing} the same stretch holds for the resulting stateless routing scheme.
	
	For the second result we invoke Algorithm \ref{alg:base_algo} with $x = n^{1/3}$ with a wider local exploration (by a factor $1/\eps$) to distance $h = h'/\eps$ whereas $h' \in \tilT(x)$ is the usual value from Lemma \ref{lem:node_sampling}. Again, the runtime follows from Lemma \ref{lem:base_algo_runtime}.

	The idea is that a wider local search allows us to compute shortest paths with at most $h$ hops exactly which allows us to disregard those for proving the stretch. By Lemma \ref{lem:algo-approx-unweighted} the additive error scales only in $h'$ which is by a factor $1/\eps$ smaller than $h$ and therefore has a relatively small impact on the stretch of paths longer than $h$ hops.
	
	More precisely, if $d_{h}(v,u) = d(v,u)$ then we obtain an exact distance estimation from Equation \eqref{eq:approx_oracle_function}.
	Now let $d_{h}(v,u) > d(v,u)$ meaning that all $v$-$u$-paths with at most $h$ hops are longer than $d(v,u)$ or there exists none such path. Then, due to the minimum edge weight of 1, we have $d(v,u) \geq h$.
	Then by Lemma \ref{lem:algo-approx-unweighted} we have
	\begin{align*}
		o_v(\lambda(u)) & \leq d(v,u)\!+\!2h' W'  = d(v,u)\!+\!2\eps h W' \\
		& \leq d(v,u)\!+\!2\eps d(v,u) = d(v,u)(1\!+\!2\eps)
	\end{align*}
	The final result is obtained with a substitution $\eps':= 2\eps$.
\end{proof}
}

\section{Node Communication Problem}

\label{sec:node_communication_problem}

\lng{In the remaining sections of this paper we will completely focus on computational lower bounds.} This section is dedicated to creating an ``information bottleneck'' in the \hybrid model between two (distant) parts of the local communication graph. We do this for the more general $\hybridpar{\infty}{\gamma}$ model, where we have a global communication bandwidth of $\gamma$ bits per node per round. Besides the advantage of having a more general lower bound, this avoids logarithmic terms and $\bigO$-notation as long as possible (recall that $\hybrid = \hybridparbig{\infty}{\bigO(\log^2 n)}$). 

We start with some intuition. Let $G=(V,E)$ be a graph. Assume that information, formalized as the state of some random variable $X$, is collectively known by nodes $A \subset V$ and must be learned by some set $B \subset V$ disjoint from $A$. This information either has to travel the hop distance $h := hop(A,B)$ from $A$ to $B$ along the unrestricted local network, which takes $h$ rounds. If we want to be faster than that, each bit traveling from $A$ to $B$ has to use a global edge eventually. Thus the number of rounds is at least the ``total amount'' of information (given by the entropy $H(X)$\footnote{The Shannon entropy of random variable $X\!:\! \Omega \!\to\! S$ is defined as $H(X) := - \!\sum_{x \in S} \mathbb{P}(X \!=\! x) \log \big(\mathbb{P}(X \!=\! x) \big)$ \cite{Shannon1948}.}), divided by the overall global communication capacity of at most $n \cdot \gamma$ bits in the $\hybridpar{\infty}{\gamma}$ model.



To prove this formally and in a more general form, we first introduce some definitions. We say that the nodes from some set $A \subseteq V$ \textit{collectively know} the state of a random variable $X$, if its state can be derived from the information that the nodes $A$ have. Or, in terms of information theory, given the state or input $S_A$ of all nodes $A$ (interpreted as a random variable), then the conditional entropy $H(X|S_A)$ also known as the amount of new information of $X$ provided that $S_A$ is already known, is zero.
Similarly, we say that the state of $X$ is \textit{unknown} to $B \subseteq V$, if the initial information of the nodes $B$ does not induce any knowledge on the outcome of $X$. Or formally, that for the state $S_B$ of the nodes $B$ we have that $H(X|S_B) = H(X)$, meaning that all information in $X$ is new even if $S_B$ is known. Another way of expressing this is that $S_B$ and $X$ are stochastically independent.


\begin{definition}[Node Communication Problem]
	\label{def:node_comm_problem}
	Let $G=(V,E)$ be some graph. Let $A,B \subset V$ be disjoint sets of nodes and $h := hop(A,B)$. Furthermore, let $X$ be a random variable whose state is collectively known by the nodes $A$ but unknown to any set of nodes disjoint from $A$. An algorithm $\mathcal A$ solves the \emph{node communication problem} if the nodes in $B$ collectively know the state of $X$ after $\mathcal A$ terminates. We say $\calA$ has success probability $p$ if $\calA$ solves the problem with probability at least $p$ for \emph{any} state $X$ can take.\footnote{In line with our Definition \ref{def:rand_algo} of success probability for graph algorithms.}
\end{definition}

The goal of Lemma \ref{lem:reduce_comm_problems} is to reduce a more basic communication problem, for which we can provide lower bounds using basic information theory (c.f., Appendix \ref{sec:information_theory}) {to} the node communication problem. Analogously to node sets, we define that Alice knows some random variable $X$, which is unknown to Bob as follows. Given that $S_{\text{Alice}}$ and $S_{\text{Bob}}$ are their respective inputs then we have $H(X|S_{\text{Alice}}) = 0$ and $H(X|S_{\text{Bob}}) = H(X)$.

\begin{definition}[Two Party Communication Problem]
	\label{def:party_comm_problem}
	Given two computationally unbounded parties, Alice and Bob, where initially Alice knows the state of some random variable $X$ which is unknown to Bob. A communication protocol $\mathcal P$ is said to solve the problem if after the execution of $\mathcal P$ Bob can derive the state of $X$ from the transcript of all exchanged messages. Performance is measured in the length of the transcript in bits. We say $\mathcal P$ has success probability $p$ if $\mathcal P$ solves the problem with probability at least $p$ for \emph{any} state $X$ can take.
\end{definition}

The reduction from the 2-party communication problem to the node communication problem uses a simulation argument similar to the one in \cite{Kuhn2020}. We show that Alice and Bob can together simulate a \hybridpar{\infty}{\gamma} model algorithm for the node communication problem and use it solve the 2-party communication problem. \shrt{The proof is deferred to Appendix \ref{sec:node_communication_problem_proofs}.}

\begin{lemma}
	\label{lem:reduce_comm_problems}
	Any algorithm $\mathcal A$ that solves the \emph{node communication problem} (Def.\ \ref{def:node_comm_problem}) in the \hybridpar{\infty}{\gamma} model on some local graph $G = (V,E)$ with $n = |V|$ and $A,B \subset V$ in $T < h = hop(A,B)$ rounds with success probability $p$ can be used to obtain a protocol $\mathcal P$ that solves the \textit{two party communication problem} (Def.\ \ref{def:party_comm_problem}) with the same success probability $p$ and transcript length at most $T \cdot n \cdot \gamma$.
\end{lemma}

\lng{\begin{proof}
	We will derive a protocol $\mathcal P$ that uses (i.e., simulates) algorithm $\mathcal A$ in order to solve the two-party communication problem. First we make a few assumptions about the initial knowledge of both parties in particular about the graph $G$ from the node communication problem, you can think of this information as hard coded into the instructions of $\mathcal P$. The important observation is that none of these assumptions will give Bob any knowledge about $X$.
			
	Specifically, we assume that Alice is given complete knowledge of the topology $G$ and inputs of all nodes in $G$ (in particular the state of $X$ and the source codes of all nodes specified by $\calA$). Bob is given the same for the subgraph induced by $V \setminus A$, which means that the state of $X$ remains unknown to Bob (c.f., Def.\ \ref{def:node_comm_problem}). To accommodate randomization of $\mathcal A$, both are given the same copy of a string of random bits (determined randomly and independently from $X$) that is sufficiently long to cover all ``coin flips'' used by any node in the execution of $\mathcal A$.
	
	Alice and Bob simulate the following nodes during the simulated execution of algorithm $\mathcal A$. For $i \in [h\!-\!1]$ let $V_i := \{v \in V \mid \hop(v,A) \leq i \}$ be the set of nodes at hop distance at most $i$ from $A$. Note that $A \subseteq V_i$ for all $i$. In round 0 of algorithm $\mathcal A$, Alice simulates all nodes in $A$ and Bob simulates all nodes in $V \setminus A$. However, in subsequent rounds $i > 0$, Alice simulates the larger set $A \cup V_i$ and Bob simulates the smaller set $B \cup V \setminus V_i$.
	
	Figuratively speaking, in round $i$ Bob will relinquish control of all nodes that are at hop distance $i$ from set $A$, to Alice. This means, in each round, every node is simulated \textit{either} by Alice \textit{or} by Bob.	
	We show that each party can simulate their nodes correctly with an induction on $i$. Initially ($i=0$), this is true as each party gets the necessary inputs of the nodes they simulate. Say we are at the beginning of round $i > 0$ and the simulation was correct so far. It suffices to show that both parties obtain all messages that are sent (in the \hybridpar{\infty}{\gamma} model) to the nodes they currently simulate.
	
	The communication taking place during execution of $\mathcal A$ in the \hybridpar{\infty}{\gamma} model is simulated as follows. If two nodes that are currently simulated by the \emph{same} party, say Alice, want to communicate, then this can be taken care as part of the internal simulation by Alice. If a node that is currently simulated (w.l.o.g.) by Bob wants to send a message over the \emph{global} network to some node that Alice simulates, then Bob sends that message directly to Alice as part of $\mathcal P$, and that message becomes part of the transcript.
	
	Now consider the case where a \textit{local} message is exchanged between some node $u$ simulated by Alice and some node $v$ simulated by Bob. Then in the subsequent round Alice will \textit{always} take control of $v$, as part of our simulation regime. Thus Alice can continue simulating $v$ correctly as she has all information to simulate all nodes all the time anyway (Alice is initially given all inputs of all nodes). Therefore it is \textit{not} required to exchange {any} \textit{local} messages across parties for the correct simulation.
	
	After $T$ simulated rounds, Bob, who simulates the set $B$ until the very end (as $T < h$), can derive the state of $X$ from the local information of $B$ with success probability at least $p$ (same as algorithm $\mathcal A$). Hence, using the global messages that were exchanged between Alice and Bob during the simulation of algorithm $\mathcal A$ we obtain a protocol $\mathcal P$ that solves the two party communication problem with probability $p$. Since total global communication is restricted by $n \cdot \gamma$ bits per round in the \hybridpar{\infty}{\gamma} model, Alice sends Bob at most $T \cdot n \cdot \gamma$ bits during the whole simulation.
\end{proof}}

Next we plug in the lower bound for the 2-party communication problem (c.f.\ Lemma \ref{lem:lower_bound_two_party} in Appendix \ref{sec:information_theory}) to derive a lower bound for the node communication problem. Note that this theorem only depends on the hop distance $h$ between $A,B$ and the entropy of $X$ and is otherwise agnostic to the local graph. Note that a lower bound that holds in expectation is also a worst case lower bound.\footnote{A worst case lower bound means there exists one outcome of $X$ where the algorithm takes at least that many rounds. A lower bound that holds in expectation clearly implies the same in the worst case.}

\begin{theorem}
	\label{thm:lower_bound_node_communication}
	Any algorithm that solves the {node communication problem} (Def.\ \ref{def:node_comm_problem}) on some $n$-node graph in the \hybridpar{\infty}{\gamma} model with success probability at least $p$, takes at least \smash{$\min\!\big(\frac{pH(X) -1}{n \cdot \gamma}, h\big)$} rounds in expectation, where $H(X)$ denotes the entropy of $X$.
\end{theorem}

\begin{proof}
	We have to show that a randomized, \hybridpar{\infty}{\gamma} algorithm $\mathcal A$ that solves the node communication problem in \textit{less} than $h$ rounds with success probability $p$ takes at least \smash{$\frac{p H(X)-1}{n \cdot \gamma}$} rounds. 
	Presume, for a contradiction, that $\mathcal A$ has an expected running time $T<h$ \textit{and} \smash{$T < \frac{p H(X)-1}{n \cdot \gamma}$}. This implies \smash{$T \cdot n \cdot \gamma < p\cdot H(X)-1$}. 
	
	Invoking Lemma \ref{lem:reduce_comm_problems} gives us a protocol $\mathcal P$ with the same success probability $p$ and with a transcript of length at most $T \cdot n \cdot \gamma$.
	With the inequality above, this means in the protocol $\mathcal P$, Alice sends \textit{less} than $p\cdot H(X)-1$ bits to Bob in expectation. This contradicts the fact that $p \cdot H(X) -1$ is a lower bound for this due to Appendix \ref{sec:information_theory} Lemma \ref{lem:lower_bound_two_party}.
\end{proof}
	
We have to accommodate the fact that in the routing problem or distance oracle problem, the nodes have to give a distance estimation or next routing neighbor only when provided with the label of the target node. Therefore we have to slightly amend Theorem \ref{thm:lower_bound_node_communication}, which will later allow us to argue that even if we assume that nodes have advance knowledge of a selection of sufficiently small labels, the lower bound will not change asymptotically. \shrt{The proof is deferred to Appendix \ref{sec:node_communication_problem_proofs}.}
	
\begin{corollary}
	\label{cor:lower_bound_general_amended}
		If $A$ is allowed to communicate $y$ bits to $B$ for free, then any algorithm that solves the {node communication problem} on some $n$-node graph (Def.\ \ref{def:node_comm_problem}) in the \hybridpar{\infty}{\gamma} model with success probability at least $p$, takes at least \smash{$\min\!\big(\frac{pH(X) -1-y}{n \cdot \gamma}, h\big)$} rounds in expectation (i.e., also in the worst case).
\end{corollary}
			
\lng{\begin{proof}
		As above, the node communication problem reduces to the communication problem between Alice and Bob, where Alice is now allowed to send $y$ bits to Bob in advance. Note that this still requires Alice to send $p \cdot H(X)-1-y$ remaining bits in expectation, as per Lemma \ref{lem:lower_bound_two_party}. 	
		The same contradiction as in Theorem \ref{thm:lower_bound_node_communication} can be derived as follows. Fewer rounds than stated in this lemma would imply that the transcript of global messages from Alice to Bob, would be shorter than $p \cdot H(X)-1-y$ bits (essentially by substituting $p \cdot H(X)-1$ for $p \cdot H(X)-1-y$ in the previous proof). Thus the transcript would be less than $p \cdot H(X)-1$ bits even when we add the $y$ ``free'' bits to the transcript.
\end{proof}}

\section{Lower Bounds For Unweighted Graphs}
\label{sec:lower_bounds_unweighted}

In this and the following section we aim to reduce from the node communication problem in the \hybridpar{\infty}{\gamma} model given in Definition \ref{def:node_comm_problem}, to the problem of computing routing tables or distance oracles, which works as follows.

We define a graph $\Gamma = (V_\Gamma, E_\Gamma)$ such that, first, the solution of the routing or distance oracle problems informs a subset $B \subset V_\Gamma$ about the exact state of some random variable $X$ that is encoded by the subgraph induced by $A \subset V_\Gamma$. Second, $X$ has a large entropy (we aim for super-linear in $n$). And third, the distance $hop(A,B)$ between both sets is sufficiently large.


\begin{definition}
	\label{def:unweighted_construction}
	Let \smash{$X = (x_{ij})_{i,j\in[k]} \in \{0,1\}^{k^2}$} be a bit sequence of length \smash{$k^2$}. Let $\Gamma = (V_\Gamma, E_\Gamma)$ (shown in Figure \ref{fig:lower_bound_basic}) be an unweighted graph with source nodes $s_1, \dots, s_k \in V_\Gamma$, transit nodes $u_1, \dots, u_k \in V_\Gamma$ and target nodes $t_1, \dots, t_k \in V_\Gamma$. 	
	Each source $s_i$ has a path of length $h$ hops to the transit nodes $u_i$. We have an edge between $u_i$ and $t_j$ if and only if $x_{ij} = 1$.
	Additionally, there are two nodes $v,v'\in V_\Gamma$ connected by a path of $h$ hops. The nodes $v$ and $v'$ have an edge to each source $s_i$ or target $t_i$, respectively (Figure \ref{fig:lower_bound_basic}).
\end{definition}

This construction has the following properties. 

\begin{enumerate}[(1)]
	\item The distance from source $s_i$ to $t_j$ is larger for $x_{ij} = 0$ than for $x_{ij} = 1$ (as shown by the subsequent Lemma \ref{lem:source_target_distance_unweighted}).
	\item For all $i,j \in [k]$, independently set $x_{ij} = 1$ with probability $\frac{1}{2}$, else $x_{ij} = 0$. This maximizes \smash{$H(X) = -k^2 \cdot  \frac{\log(1/2)}{2} = \frac{k^2}{2}$}.
	\item Let $A \!=\!\{u_1, ... \,, u_k,t_1, ... \,, t_k\}$, $B \!=\! \{s_1, ... \,, s_k\}$, i.e., $hop(A,B) \!=\! h$.	
\end{enumerate}

\lng{\begin{figure}[h]
	\centering
	\includegraphics[scale=0.7]{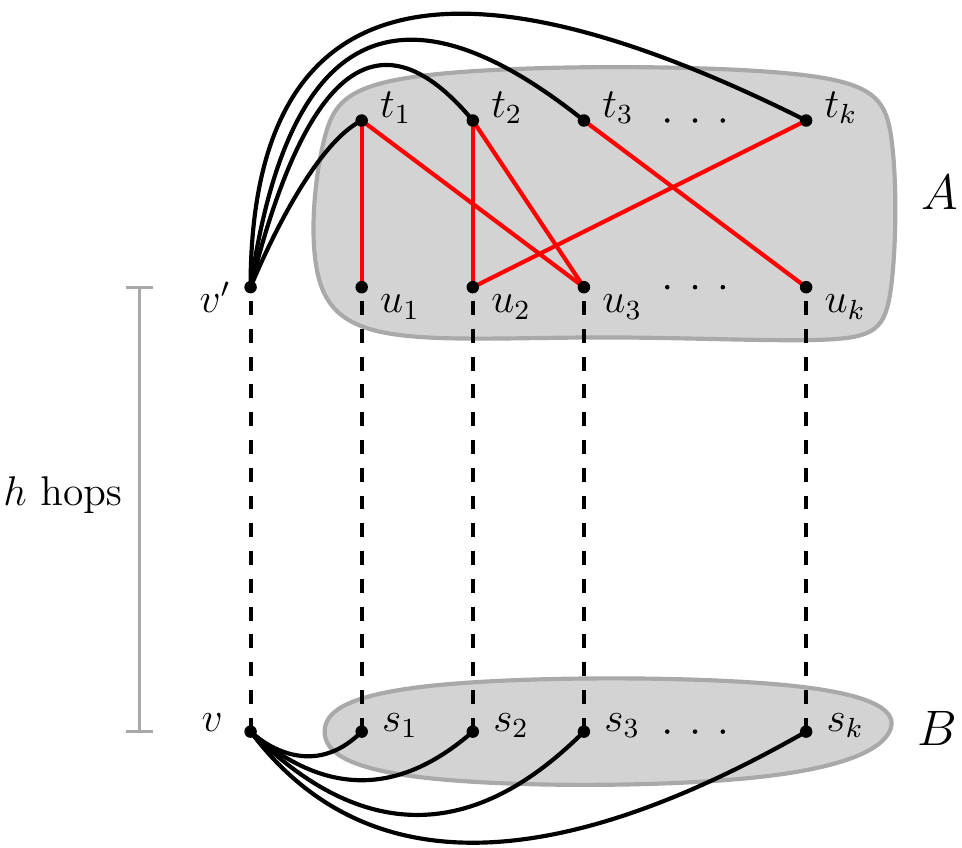}
	\caption{\lng{\boldmath} Graph $\Gamma = (V_\Gamma,E_\Gamma)$. Bit string $X = (x_{ij})$ determines red edges. E.g., $\{u_1, t_2\} \!\notin\! E_\Gamma$ and $\{u_1, t_3\} \!\in\! E_\Gamma$ means $x_{12} = 0$ and $x_{13} = 1$, respectively.}
	\label{fig:lower_bound_basic}
\end{figure}}

\shrt{\begin{figure}[h]
		\centering
		\begin{subfigure}{0.48\textwidth}
			\includegraphics[scale=0.73]{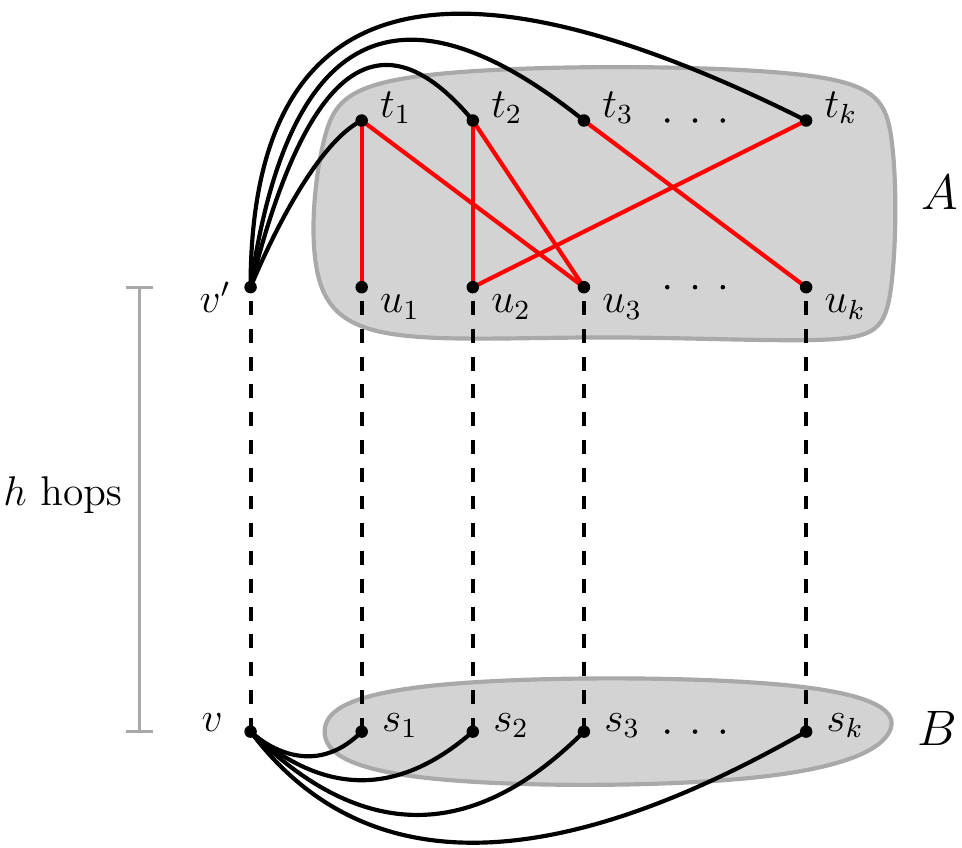}
			\caption{\lng{\boldmath} Graph $\Gamma = (V_\Gamma,E_\Gamma)$. Bit string $X = (x_{ij})$ determines red edges. E.g., $\{u_1, t_2\} \!\notin\! E_\Gamma$ and $\{u_1, t_3\} \!\in\! E_\Gamma$ means $x_{12} = 0$ and $x_{13} = 1$, respectively.}
			\label{fig:lower_bound_basic}
		\end{subfigure}
		\hfill
		\begin{subfigure}{0.48\textwidth}
			\includegraphics[scale=0.73, page = 2]{img/lower_bound.pdf}
			\caption{\lng{\boldmath} Graph $\Gamma$ constructed around $G = (A,E) \in \mathcal G_{k,\ell}$ with $m := |E|$ edges (in red) and from the bit string $X = (x_e)_{e \in E} \in \{0,1\}^m$. An edge $e \in E$ is part of $\Gamma$ iff $x_e = 1$. 
			}
			\label{fig:lower_bound_enhanced}
		\end{subfigure}
		\vspace*{-2mm}
		\caption{Lower bound graphs. Unweighted (left) and weighted (right).}
\end{figure}}
\vspace*{-2mm}

\lng{Property (1) is shown by the following lemma.}

\begin{lemma}
	\label{lem:source_target_distance_unweighted}
	{If $x_{ij}\!=\!1$ then $d(s_i,t_j) = h\!+\!1$ and the shortest $s_i$-$t_j$-path contains $v$, else $d(s_i,t_j) = h\!+\!2$ and it does not contain $v$.}\footnote{In the following, distances without subscript will refer to distances in $\Gamma$, i.e., $d(u,v) := d_\Gamma(u,v)$.}
\end{lemma}

\begin{proof}
	 Any path from $s_i$ to $t_j$ has to cross the vertex cut $U := \{u_1, \dots ,u_k,v'\}$ (c.f., Figure \ref{fig:lower_bound_basic}). Such a path has to include a path of length $h$ to reach a node of $U$, as well as an additional edge connecting $U$ to $t_j$ and therefore $d(s_i,t_j) \geq h+1$. However, we also have $d(s_i,t_j) \leq h+2$, due to the path along the nodes $s_i, v, \ldots, v',t_j$ (c.f., Figure \ref{fig:lower_bound_basic}) that has length $h+2$.
	
	If $x_{ij} = 1$, i.e., $\{u_i, t_j\} \in E$, then the path along the nodes $s_i, \ldots, u_i,t_j$ has length $h+1$. Note that all nodes in $U \setminus \{u_i,v'\}$ are at distance at least $h+2$ from $s_i$ (c.f., Figure \ref{fig:lower_bound_basic}), so every path via one of the nodes $U \setminus \{u_i,v'\}$ has distance at least $h+3$. In the case $x_{ij} = 0$, i.e., $\{u_i, t_j\} \notin E_\Gamma$, this is also true for the path via $u_i$ and the only path with distance $h+2$ is the one via $v'$.
\end{proof}

The idea to prove the next theorem is that if the nodes in $B$ learn the distance to the nodes in $t_1, \dots, t_k$, then their combined knowledge can be used to infer the state of the random string $X$ that is collectively known by the nodes in $A$.

\begin{theorem}
	\label{thm:lower_bound_unweighted}
	Even on \emph{unweighted graphs}, any randomized algorithm that computes exact ({stateless} or stateful) routing schemes or distance oracles in the \hybridpar{\infty}{\gamma} model with constant success probability takes \smash{$\Omega(n^{1/3}/\gamma^{1/3})$} rounds. This holds for labels of size up to \smash{$c \cdot n^{2/3} \cdot \gamma^{1/3}$} (for a fixed constant $c > 0$).
\end{theorem}

\begin{proof}
	Consider an algorithm $\mathcal A$ that computes exact routing schemes or distance oracles for $\Gamma$ with constant probability at least $p$.
	Assuming that the nodes in $B$ are given the distance labels for the nodes $\{t_1, \dots ,t_k\}$ in advance then after algorithm $\mathcal A$ has terminated, every node $s_i \in  B$ either knows its distance $d(s_i,t_j)$ for every $t_j$ (distance oracles) or it knows if the next node on the corresponding shortest path is $v$ or not (routing schemes).
	
	Then, by Lemma \ref{lem:source_target_distance_unweighted}, every node $s_i$ can infer for every $t_j$ if $x_{ij} = 0$ or $x_{ij} = 1$, which means that the nodes in $B$ collectively know the state of $X$. Note that this corresponds to the node communication problem of Definition \ref{def:node_comm_problem}, with the caveat that we assumed nodes in $B$ have advance knowledge of the labels $\lambda(t_1), \dots , \lambda(t_k)$.
	The combination of these labels can contain information at most \smash{$y := \sum_{i =1}^k |\lambda(t_i)|$} bits, which we allow $A$ to communicate to $B$ ``for free''. By Corollary \ref{cor:lower_bound_general_amended}, algorithm $\mathcal A$ takes at least \smash{$\min\!\big(\frac{pH(X)-1-y}{n \cdot \gamma}, h\big)$} rounds.

	The total number of nodes of $\Gamma$ is $n \in \Theta(h\cdot k)$, which leaves one degree of freedom for $k$ and $h$. We choose \smash{$k \in \Theta(n^{2/3}\gamma^{1/3})$} and \smash{$h \in \Theta(n^{1/3}/\gamma^{1/3})$}. This implies \smash{$H(X)\in \Theta(n^{4/3}\gamma^{2/3})$} by property (2) (see further above). For labels of size \smash{$c \cdot n^{2/3} \gamma^{1/3}$}, we have \smash{$y =  c \cdot \Theta(n^{4/3}\gamma^{2/3})$}. We simply choose the constant $c > 0$ small enough so that \smash{$c n^{4/3}\gamma^{2/3} \leq  \frac{pH(X)-1}{2}$} (i.e., $y$ will not change the lower bound asymptotically). Plugging all of the above into the lower bound of \smash{$\min\!\big(\frac{pH(X)-1-y}{n \cdot \gamma}, h\big)$} rounds, yields the desired bound of \smash{$\Omega(n^{1/3}/\gamma^{1/3})$} rounds.
\end{proof}

\section{Lower Bounds for Approximations}
\label{sec:lower_bounds_approx}

Our next construction relies on the existence of families of graphs that have high girth and maintain relatively high density.
We modify the basic construction above, essentially by replacing the upper part of $\Gamma$ with a random selection of edges from a graph of that family (and also making $\Gamma$ weighted). Besides high density we require the following.

\begin{definition} 
	\label{def:graph_family}
	$\mathcal G_{k,\ell}$ is a graph family, s.t. for all $G \!=\! (A,E) \!\in\! \mathcal G_{k,\ell}$
	\begin{enumerate}[(i)]
		\item $|A| = 2k$
		\item $G$ has (even) girth at least $\ell$		
		\item $G$ is balanced and bipartite		
	\end{enumerate}	
\end{definition}

\subsection{Weighted Construction}

Removing an edge from $G \in \mathcal G_{k,\ell}$ incurs a large detour of at least $\ell\!-\!1$ hops between the endpoints of that edge, since otherwise there would be a cycle shorter than $\ell$ in $G$. This observation is often used to prove certain bounds for low stretch subgraphs\footnote{One prominent example is the lower bound on the size of low stretch spanners.} and can be exploited to introduce a stretch into our lower bound construction. We construct this formally as follows (however, first consulting Figure \ref{fig:lower_bound_enhanced} will presumably be more helpful to the reader).

\lng{\begin{figure}[h]
	\centering
	\includegraphics[scale=0.7, page = 2]{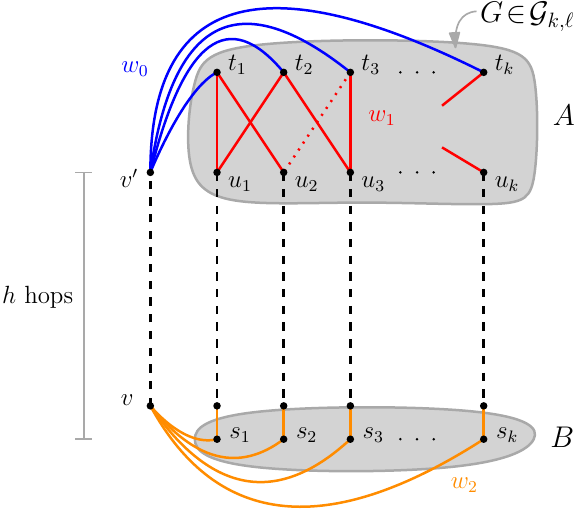}
	\caption{\lng{\boldmath} Graph $\Gamma$ constructed around $G = (A,E) \in \mathcal G_{k,\ell}$ with $m := |E|$ edges (in red) and from the bit string $X = (x_e)_{e \in E} \in \{0,1\}^m$. An edge $e \in E$ is part of $\Gamma$ iff $x_e = 1$. For instance $\{u_2,t_3\} \in E$ but $x_{\{u_2,t_3\}} = 0$, so $\{u_2,t_3\} \notin E_\Gamma$. Edge weights $ \textcolor{blue}{w_0},\textcolor{red}{w_1}, \textcolor{orange}{w_2}$ are indicated with colored edges, all others have weight 1.}
	\label{fig:lower_bound_enhanced}
\end{figure}}

\begin{definition}
	\label{def:weighted_construction}
Let $G = (A,E) \in \mathcal G_{k,\ell}$ with $m := |E|$ edges and let $\{u_1, \dots, u_k\} \cup \{t_1, \dots, t_k\} = A$ be the bipartition of $G$. Graph $\Gamma = (V_\Gamma,E_\Gamma)$ (shown in Figure \ref{fig:lower_bound_enhanced}) has a similar structure as the unweighted construction (Def.\ \ref{def:unweighted_construction}), where the main difference is the way how the nodes $\{u_1, \dots, u_k\} \cup \{t_1, \dots, t_k\}$ are connected by edges in $\Gamma$.

Let $X = (x_{e})_{e\in E} \in \{0,1\}^{m}$ be a bit string of length $m = |E|$, i.e., each bit $x_e$ corresponds to an edge $e$ of $G$. For each $u_i,t_j$ we have $\{u_i,t_j\} \in E_\Gamma$, if and only if $\{u_i,t_j\} \in E$ \textit{and} $x_{\{u_i,t_j\}} = 1$. In a slight change from the previous construction, we make the path from $v$ to $v'$ of hop length $h\!-\!1$. The weights of $\Gamma$ are assigned as follows. Edges between the node $v'$ and some $t_j$ have weight $w_0$. Edges between nodes $u_i,t_j$ have weight $w_1$. Edges incident to some $s_i$ have weight $w_2$.
\end{definition}

We have the following properties.

\begin{enumerate}[(1)]
	\item Let $e \!=\! \{u_i,t_j\} \!\in\! E$. $w_0,w_1,w_1$ can be chosen s.t.\ $d(s_i,t_j)$ is much longer for $x_e = 0$ than for $x_e = 1$ (c.f., Lemma \ref{lem:source_target_distance_weighted}).
	\item For each edge $e \in E$ of $G$, set $x_{e} = 1$ i.i.d.\ with probability $\frac{1}{2}$, else $x_{e} = 0$. This maximizes the entropy $H(X) = \frac{m}{2}$.
	\item For nodes $A$ of $G$ and $B := \{s_1, \dots, s_k\}$ we have $\hop(A,B) = h$.
\end{enumerate}

%

We analyze distances $d(s_i,t_j)$ between nodes $s_i,t_j$ with $e = \{u_i, t_j\} \in E$ for the two cases that $e$ is part of $\Gamma$ ($x_e = 1$), or not ($x_{e} = 0$). Conceptually, we choose weights $w_1 \ll w_0$, such that we can observe a large difference in $d(s_i,t_j)$ depending on $x_e$. \shrt{The proof is deferred to Appendix \ref{sec:lower_bounds_approx_proofs}.} 

\begin{lemma}
	\label{lem:source_target_distance_weighted}
	Consider $\Gamma$ (Def.\ \ref{def:weighted_construction}), constructed from $G = (A,E) \in \mathcal G_{k,\ell}$ and $X$. Let $w_1 < w_0 < (\ell-1) w_1$. Let $e = \{u_i,t_j\} \in E$. Then we have:
	\begin{enumerate}[(i)]
		\item The shortest $s_i$-$t_j$-path contains $v$ if and only if $x_e = 0$.
		\item If $x_{e}=1$ then $d(s_i,t_j) = w_2 + w_1 + h - 1$.\\ If $x_{e}=0$, then $d(s_i,t_j) = w_2 + w_0 + h - 1$.
	\end{enumerate}
	 
\end{lemma}

\lng{\begin{proof}
	Let $U := \{u_1, \dots ,u_k,v\}$ be the vertex cut that separates any $s_i$ from any $t_j$. The shortest \textit{simple} $s_i$-$t_j$-path that crosses $U$ via $v$ has length $w_2 + w_0 + h - 1$ \textit{independently} from $x_e$ (simple implies that a path can not ``turn around'' and go via $u_i$). 
	
	Consider the shortest $s_i$-$t_j$-path that does \textit{not} contain $v$. In the case $x_e=1$, i.e., $e = \{u_i,t_j\}$ exists in $\Gamma$, this $s_i$-$t_j$-path is forced to cross $U$ via $u_i$ and then goes directly to $t_j$ via $e$, and thus has length $w_2 + w_1 + h - 1$.
	
	Let us analyze the length of the $s_i$-$t_j$-path that does \textit{not} contain $v$ for the case $x_e=0$ (i.e., $e \notin E_\Gamma$). Let $G'$ be the subgraph that corresponds to $G$ after removing each edge $e' \in E$ with $x_{e'}=0$. Then that $s_i$-$t_j$-path has to traverse $G'$ to reach $t_j$.	
	The sub-path from $u_i$ to $t_j$ in $G'$ has to use at least $(\ell\!-\!1)$ edges, because otherwise $e = \{u_i,t_j\}$ would close a loop of less than $\ell$ edges in $G'$ (and thus also in $G$), contradicting the premise that $G$ has girth $\ell$. Thus, for $e \notin E_\Gamma$ any $s_i$-$t_j$-path that does not contain $v$ has length \textit{at least} $w_2+(\ell\!-\!1)w_1+h-1$.
	
	We sum up the cases. If $x_e = 1$, then the $s_i$-$t_j$-path \textit{not} containing $v$ of length $w_2 + w_1 + h - 1$ is shorter than the one via $v$ of length $w_2 + w_0 + h - 1$, since $w_1 < w_0$. If $x_e = 0$, then the $s_i$-$t_j$-path via $v$ of length $w_2 + w_0 + h - 1$ is shorter than the one not containing $v$ of length at least $w_2 + (\ell-1)w_1 + h - 1$ due to $w_0 < (\ell\!-\!1) w_1$.
\end{proof}}

For the reduction from the node communication problem to our concrete routing and distance oracle problems, we start with a technical lemma that analyzes the running time of any algorithm $\mathcal A$ that solves the node communication problem in $\Gamma$ for the dedicated node sets $A,B$ and the random variable $X$ from which $\Gamma$ is constructed.

In particular, we express the lower bound from Theorem \ref{thm:lower_bound_node_communication} as function of $n := |V_\Gamma|$, the density of $G \in \mathcal G_{k,\ell}$ given by a parameter $\delta$ and the global communication capacity $\gamma$. 
The lemma is the result of balancing a trade off between the distance $h = \hop(A,B)$ and the number of nodes $\Theta(k)$ of $G$ (which governs the entropy $H(X)=\Theta(k^{1+\delta})$ when the density of $G$ is fixed).\footnote{Naturally, the number of edges $\Theta\big(k^{1+\delta}\big)$ of $G \in \mathcal G_{k,\ell}$ is also bounded through $\delta < \frac{\ell-1}{2}$, c.f., Lemma \ref{lem:girth_edges_bound} in Appendix \ref{sec:girth_density}.} \shrt{The proof is deferred to Appendix \ref{sec:lower_bounds_approx_proofs}.}

\begin{lemma}
	\label{lem:approx_lower_bound_groundwork}
	Consider $\Gamma$ constructed from random variable $X$ and $G = (A,E) \in \mathcal G_{k,\ell}$ (Def.\ \ref{def:weighted_construction}) with \smash{$|E| = \Theta\big(k^{1+\delta}\big)$} edges (for $\delta > 0$ and $k$ of our choosing).
	Let $\mathcal A$ be an algorithm that solves the node communication problem on $\Gamma$ with $X$, node sets $A,B \subset V_\Gamma$ and $h = hop(A,B)$ in the \hybridpar{\infty}{\gamma} model (all parameters as in Def.\ \ref{def:weighted_construction}). We can choose $k=\Theta\big(\frac{n}{h}\big)$ such that $\mathcal A$ takes \smash{$\Omega\big((\tfrac{n^\delta}{\gamma})^{\frac{1}{2+\delta}}\big)$} rounds. There exists a constant $c>0$ such that this holds even when we allow exchanging \smash{$c \cdot k^{1+\delta}$} bits from $A$ to $B$ for free.
\end{lemma}

\lng{\begin{proof}
	As $\mathcal A$ solves the node communication problem (Def.\ \ref{def:node_comm_problem}) it takes at least \smash{$\min\!\big(\frac{pH(X)-1-y}{n \cdot \gamma}, h\big)$} rounds by Corollary \ref{cor:lower_bound_general_amended}, where $H(X) \in \Theta (k^{1+\delta})$ (see property (2), further above), $p$ is the constant success probability and $y$ describes the ``free'' communication. 
	
	The arguments of \smash{$\min\!\big(\frac{pH(X)-1-y}{n \cdot \gamma}, h\big)$} behave inversely, since increasing the distance $h = hop(A,B)$ leaves only \smash{$k \in \Theta\big(\frac{n}{h}\big)$} nodes for the graph $G$, which decreases $H(X) \in \Theta (k^{1+\delta})$. So in order to maximize the number of rounds given by the min function, we solve the equation \smash{$\frac{pH(X)-1-y}{n \cdot \gamma} = h$} subject to $k \cdot h = \Theta(n)$. Slashing constants and neglecting $y$ for now, this simplifies as follows
	\[
	\Theta \big( \tfrac{k^{1+\delta}}{n\cdot \gamma} \big) = \Theta (h), \quad \text{subject to} \quad k \cdot h = \Theta(n).
	\] 
	The solution is \smash{$k = \Theta\big(n^{\frac{2}{2+\delta}}\cdot \gamma^{\frac{1}{2+\delta}}\big)$} and \smash{$h = \Theta\big(n^{\frac{\delta}{2+\delta}}/ \gamma^{\frac{1}{2+\delta}}\big)$} (which the willing reader may verify by inserting), resulting in a lower bound of \smash{$\Omega(n^{\frac{\delta}{2+\delta}} / \gamma^{\frac{1}{2+\delta}}) = \Omega\big((\tfrac{n^\delta}{\gamma})^{\frac{1}{2+\delta}}\big)$} rounds.
	
	When we factor the free communication of $y = c \cdot k^{1+\delta}$ bits back into the equation \smash{$\frac{pH(X)-1-y}{n \cdot \gamma} = h$}, then there are no asymptotic changes to the outcome of our calculations as long as we choose the constant $c$ such that \smash{$c \cdot k^{1+\delta} \leq \frac{pH(X)-1}{2}$} (for all $n$ bigger than some constant $n_0$).
\end{proof}}

\subsection{Distance Oracles}



The first lower bound with stretch is for the distance oracle problem. 
The idea is as follows. In case there is a direct edge $e = \{u_i, t_j\}$ (i.e., $x_e = 1$), the distance from $s_i$ to $t_j$ is almost $\ell\!-\!1$ times shorter, than if that is not the case. Hence, by learning an approximation of $d(s_i,t_j)$ with a stretch slightly lower than $\ell\!-\!1$, the node $s_i$ can conclude if $e$ exists or not, i.e., if $x_e = 1$ or $x_e = 0$. Hence the nodes $B = \{s_1, \dots, s_k\}$ collectively learn the random variable $X$ and thus solve the node communication problem.\footnote{Note that for the distance oracles lower bound, the path from $v$ to $v'$ could be removed from $\Gamma$, since $s_i$ learns $x_e = x_{\{s_i,t_j\}}$ from the distance estimate to $t_j$ and not from the next routing node. Since it does not hurt either, we keep $\Gamma$ uniform for all our lower bounds.}

This lemma is kept general such that we can later plug in any graph with a given density parameter $\delta$ and girth $\ell$. Note that the girth $\ell$ fundamentally limits the density parameter $\delta$; the correspondence between the two is roughly $\delta \in \bigO\big(\frac{1}{\ell}\big)$ as shown in Appendix \ref{sec:girth_density}. For a more intuitive understanding we suggest plugging in the complete bipartite graph $G_{k,k}$ which has girth $\ell=4$ and $\Theta(k^2)$ edges (i.e., density parameter $\delta=1$). \shrt{The proof is deferred to Appendix \ref{sec:lower_bounds_approx_proofs}.}

\begin{lemma}
	\label{lem:lower_bound_distance_oracle}
	Consider $\Gamma$ constructed from $G = (A,E) \in \mathcal G_{k,\ell}$ with $|E| = \Theta\big(k^{1+\delta}\big)$ edges for some $\delta > 0$. 
	Any algorithm that solves the \emph{distance oracle problem} on $\Gamma$ with stretch $\alpha_\ell = \ell\!-\!1\!-\!\varepsilon$ (for any const.\ $\eps > 0$) and constant success probability in the \hybridpar{\infty}{\gamma} model takes \smash{$\Omega\big((\tfrac{n^\delta}{\gamma})^{\frac{1}{2+\delta}}\big)$} rounds, for labels up to size 
	\smash{$c \cdot n^{\frac{2\delta}{2+\delta}}\cdot \gamma^{\frac{\delta}{2+\delta}}$} (for a fixed const.\ $c > 0$).
\end{lemma}

\lng{\begin{proof}[Proof of Lemma \ref{lem:lower_bound_distance_oracle}]
		Set $w_2 = 1$ (this weight is only needed later for stateful routing scheme lower bounds). To make the idea described above work, we have to make the unweighted edges of $\Gamma$ (the $h$-hop $s_i$-$u_i$-paths) insignificant for the approximation ratio by scaling the weights $w_1,w_0$ by a large factor $t$.
		For that purpose we introduce the parameter $t>0$, which is specified later. 	
		We choose $w_1 = t$ and $w_0 =(\ell\!-\!1\!-\!\frac{\varepsilon}{2})\cdot t$. In particular, this means $w_1 < w_0 < (\ell-1) w_1$ (the precondition of Lemma \ref{lem:source_target_distance_weighted}).
		
		Let $e = \{u_i,t_j\} \in E$. Let $d_1 = d(s_i,t_j)$ for the case $x_e = 1$ and $d_0 = d(s_i,t_j)$ for the case $x_e = 0$. By Lemma \ref{lem:source_target_distance_weighted}, we know that $d_1 = t \!+\! h$ and \smash{$d_0 = (\ell\!-\!1\!-\!\frac{\varepsilon}{2})\cdot t + h$}.		
		Let $\mathcal A$ be an approximation algorithm for the distance oracle problem with stretch $\alpha_\ell = \ell\!-\!1\!-\!\varepsilon$. 
		
		For the cases $x_e = 1$ and $x_e = 0$, respectively, let $\tilde d_1$ and $\tilde d_0$ be distance approximations of $d(s_i,t_j)$ with stretch $\alpha_\ell$ that $s_i$ determines with its local table and the label of $t_j$ (which we computed with $\mathcal A$). 
		Note that our claims about $\tilde d_0, \tilde d_1$ will only depend on $x_e$ and are independent from $x_{e'}$ of other edges $e'\in E \setminus \{e\}$ (even though the \textit{exact} value of $\tilde d_0, \tilde d_1$ might depend on the $x_{e'}$).
		
		The goal is to show 
		$\tilde d_1 < d \leq \tilde d_0$ for some constant $d > 0$, which enables $s_i$ to distinguish $x_e=1$ from $x_e=0$ from its approximation of $d(s_i,t_j)$.
		We know that $\tilde d_1 \leq \alpha_\ell \cdot d_1 = (\ell\!-\!1\!-\!\varepsilon)(t \!+\! h)$. We also have $\tilde d_0 \geq d_0$ since our approximations are supposed to be one-sided.
		Then 
		\begin{align*}
		\tilde d_1 & \leq (\ell\!-\!1\!-\!\varepsilon)(t \!+\! h)\\ 		
		& = (\ell\!-\!1\!-\!\tfrac{\varepsilon}{2})\cdot t + h + (\ell\!-\!2\!-\!\varepsilon) \cdot h -\tfrac{\varepsilon}{2}\cdot t \tag*{\small\text{\textit{expand}}}\\		
		& < (\ell\!-\!1\!-\!\tfrac{\varepsilon}{2})\cdot t + h \tag*{\small\text{\textit{for large enough} $t$}}\\
		& = d_0 \leq \tilde d_0.
		\end{align*}
		
		The strict inequality is obtained by choosing $t > \frac{2(\ell-2-\varepsilon)}{\varepsilon} \cdot h \in \Theta(\frac{\ell h}{\varepsilon})$. Note that edge weights remain polynomial in $n$ with this choice of $t$ since $\ell, h \in O(n)$ and $\varepsilon$ is constant. 	
		So we get $\tilde d_1 < d_0 \leq \tilde d_0$, which implies the following. Let \smash{$\tilde d(s_i,t_j)$} be the distance estimate that $s_i$ actually outputs. Then it is $x_e=0$ if \smash{$\tilde d(s_i,t_j) < d_0$}, else it is $x_e=1$. Hence the nodes $B := \{s_1, \dots, s_k\}$ collectively learn $X$ and thus solve the node communication problem, which takes \smash{$\Omega\big((\tfrac{n^\delta}{\gamma})^{\frac{1}{2+\delta}}\big)$} rounds by Lemma \ref{lem:approx_lower_bound_groundwork}.
		
		Two things remain to be mentioned. First, we can assume that each $s_i$ also has advance knowledge of $d_0$, as this does not carry any information about $X$ to the nodes $B = \{s_1, \dots, s_k\}$ and therefore does not make the node communication problem easier. 
		
		Second, the nodes $s_i$ can only produce the distance estimations \smash{$\tilde d(s_i,t_j)$} when they are also provided with the labels $\lambda(t_1), \dots, \lambda(t_k)$.
		Here we assume that the nodes $s_i$ get these labels in advance as part of the contingent of ``free'' communication that we budgeted for in Lemma \ref{lem:approx_lower_bound_groundwork} and which does not make the node communication problem asymptotically easier. 
		
		In particular let $c_1>0$ be the constant from Lemma \ref{lem:approx_lower_bound_groundwork}, such that we are allowed \smash{$y = c_1 k^{1+\delta}$} bits of free communication in total. This leaves \smash{$c_1 k^{\delta}$} bits for each of the $k$ labels $\lambda(t_1), \dots, \lambda(t_k)$. In the proof of Lemma \ref{lem:approx_lower_bound_groundwork} we chose \smash{$k = c_2 \cdot n^{\frac{2}{2+\delta}}\cdot \gamma^{\frac{1}{2+\delta}}$} (for some constant $c_2 > 0$), resulting in a label size of at most \smash{$c_1 c_2 \cdot n^{\frac{2\delta}{2+\delta}}\cdot \gamma^{\frac{\delta}{2+\delta}}$}.
\end{proof}}


It remains to insert graphs $G \in \calG_{k,\ell}$ into Lemma \ref{lem:lower_bound_distance_oracle}. We aim for graphs $G = (V,E) \in \calG_{k,\ell}$ with $|E| \in \Theta\big(k^{1+\delta}\big)$ that maximize both girth $\ell$ and density parameter $\delta$. As outlined in Appendix \ref{sec:girth_density}, these are opposing objectives, and for even girth $\ell \geq 4$ we know that $\delta \in \bigO\big(\frac{2}{\ell-2}\big)$ (from applying Lemma \ref{lem:girth_edges_bound} on uneven girth $\ell-1$). Bipartite graphs of girth $\ell$ that reach $\delta \in \Theta\big(\frac{2}{\ell-2}\big)$ can be constructed for small girth $\ell$ (summarized in Lemma \ref{lem:low_even_girth_graph_density}) from which we obtain Theorem \ref{thm:lower_bound_distance_oracle_small_stretch}. But for higher girth we have to settle for $\delta$ below this threshold (see Lemma \ref{lem:high_even_girth_graph_density}), this is reflected in Theorem \ref{thm:lower_bound_distance_oracle_large_stretch}.

\begin{theorem}
	\label{thm:lower_bound_distance_oracle_small_stretch}
	Any algorithm that solves the \emph{distance oracle problem} in the \hybridpar{\infty}{\gamma} model with constant success probability with
	\begin{itemize}
		\item stretch $3\!-\!\eps$ takes \smash{$\Omega\big((\tfrac{n}{\gamma})^{\frac{1}{3}}\big)$} rounds for label size \smash{$\leq c \cdot (n^2 \gamma)^{\frac{1}{3}}$}
		\item stretch $5\!-\!\eps$ takes \smash{$\Omega\big(\tfrac{n^{1/5}}{\gamma^{2/5}}\big)$} 
		rounds for label size \smash{$\leq c \cdot (n^2 \gamma)^{\frac{1}{5}}$}
		\item stretch $7\!-\!\eps$ takes
		\smash{$\Omega\big(\tfrac{n^{1/7}}{\gamma^{3/7}}\big)$}
		rounds for label size \smash{$\leq c \cdot (n^2 \gamma)^{\frac{1}{7}}$}
		\item stretch $11\!-\!\eps$ takes
		\smash{$\Omega\big(\tfrac{n^{1/11}}{\gamma^{5/11}}\big)$}
		rounds for label size \smash{$\leq c \cdot (n^2 \gamma)^{\frac{1}{11}}$}
	\end{itemize}
	for any const.\ $\eps > 0$ and a fixed const.\ $c > 0$.
\end{theorem}

\begin{proof}
	By Lemma \ref{lem:low_even_girth_graph_density} there are bipartite, balanced graphs with girth $\ell \in \{4,6,8,12\}$ and \smash{$\Theta(n^{1+\frac{2}{\ell-2}})$} edges, thus $\delta(\ell) = \frac{2}{\ell-2}$. In particular, we have \smash{$\delta(4) = 1, \delta(6) = \frac{1}{2}, \delta(8) = \frac{1}{3}, \delta(12) = \frac{1}{5}$}, which yield the desired results when plugged into Lemma \ref{lem:lower_bound_distance_oracle}.
\end{proof}

Applying Lemma \ref{lem:lower_bound_distance_oracle} on the densest known graphs with larger girth (see Lemma \ref{lem:high_even_girth_graph_density}), we obtain the subsequent theorem. The parametrization is complex due to a case distinction in Lemma \ref{lem:high_even_girth_graph_density}, the upshot is that for any constant stretch and sufficiently small $\gamma$ we get lower bounds polynomial in $n$ for label up to sizes that are also polynomial.

\begin{theorem}
	\label{thm:lower_bound_distance_oracle_large_stretch}
	Any algorithm that solves the \emph{distance oracle problem} in the \hybridpar{\infty}{\gamma} model with constant success probability with
	\begin{itemize}
		\item stretch $\ell \!-\! 1 \!-\!\eps$ for $\ell \geq 14$ with $\ell \equiv 2 \!\!\mod 4$ takes
		\smash{$\Omega\Big(n^{\tfrac{2}{3\ell - 8}}/\gamma^{\tfrac{3\ell -10}{6\ell - 6}}\Big)$} rounds for label size \smash{$\leq c \cdot n^{{4}/(3\ell - 8)} \cdot \gamma^{{2}/(3\ell - 8)}$}
		\item stretch $\ell \!-\! 1 \!-\!\eps$ for $\ell \geq 16$ with $\ell \equiv 0 \!\!\mod 4$ takes
		\smash{$\Omega\Big(n^{\tfrac{2}{3\ell - 10}}/\gamma^{\tfrac{3\ell -12}{6\ell - 8}}\Big)$} rounds for label size \smash{$\leq c \cdot n^{{4}/(3\ell - 10)} \cdot \gamma^{{2}/(3\ell - 10)}$}
	\end{itemize}
	for any const.\ $\eps > 0$ and a fixed const.\ $c > 0$.
\end{theorem}

\begin{proof}
	By Lemma \ref{lem:high_even_girth_graph_density} there are bipartite, balanced graphs with even girth $\ell \geq 14$ that have  (i) \smash{$\Theta(n^{1+\frac{4}{3\ell-10}})$} edges if $\ell \equiv 2 \!\!\mod 4$, or (ii) \smash{$\Theta(n^{1+\frac{4}{3\ell-12}})$} edges if $\ell \equiv 0 \!\!\mod 4$. Thus in case (i) we have \smash{$\delta(\ell) = \frac{4}{3\ell-10}$} and in case (ii) \smash{$\delta(\ell) = \frac{4}{3\ell-12}$}. Plugging $\delta(\ell)$ into Lemma \ref{lem:lower_bound_distance_oracle} gives the desired result.
\end{proof}

\subsection{Stateless Routing Scheme}

For lower bounds of routing schemes we exploit the observation that for an edge $e = \{s_i,t_j\} \in E$ the node $s_i$ learns about the existence of $e$ in $\Gamma$, i.e., whether $x_e = 0$ or $x_e = 1$, from the decision to send a packet with destination $t_j$ first to $v$ or not. More precisely, our goal is to show that $x_e = 0$ if and only if $v$ is the first routing neighbor for the packet with destination $t_j$.

However, we have to decrease the stretch of our lower bound in order that this works. The main obstacle is that the decision of $s_i$ to send a packet with target $t_j$ directly towards $u_i$ instead of node $v$ (left path) does not impact the distance of the routing path that one can still obtain by that much.

In particular, in the case of \textit{stateless} routing, a packet that travels from $s_i$ to $u_i$ and finds that the direct edge $\{u_i,t_j\}$ does \textit{not} exist, could still use any other edge $\{u_i,t_p\}$ and then the two edges $\{t_p, v'\}, \{v', t_j\}$ to get to $t_j$ (e.g., in Figure \ref{fig:lower_bound_enhanced} from $s_2$ to $t_3$). This would mislead $s_i$ as the first routing node was \textit{not} $v$, yet $x_e = 0$.

The target is to prohibit this and some other troublesome routing options by making them exceed the stretch guarantee. However, this gives us additional restrictions that dominate the resulting system of inequalities for higher girth $\ell$ of $G$, in particular we gain no improvement in the stretch for $\ell \geq 8$. \shrt{The proof is deferred to Appendix \ref{sec:lower_bounds_approx_proofs}.}


\begin{lemma}
	\label{lem:lower_bound_stateless_routing}
	Consider $\Gamma$ constructed from $G = (A,E) \in \mathcal G_{k,\ell}$ with $|E| = \Theta\big(k^{1+\delta}\big)$ edges for some $\delta > 0$. For any constant $\varepsilon > 0$ let \smash{$\alpha_\ell = \sqrt{\ell\!-\!1} \!-\! \eps$} for $\ell \leq 6$ and \smash{$\alpha_\ell = 1 \!+\! \sqrt{2} \!-\! \eps$} for $\ell \geq 8$.
	Any algorithm that computes a \emph{stateless routing scheme} on $\Gamma$ with stretch $\alpha_\ell$ and constant success probability in the \hybridpar{\infty}{\gamma} model takes \smash{$\Omega\big((\tfrac{n^\delta}{\gamma})^{\frac{1}{2+\delta}}\big)$} rounds. This holds for labels of size 
	\smash{$c \cdot n^{\frac{2\delta}{2+\delta}} \gamma^{\frac{\delta}{2+\delta}}$} and fixed constant $c>0$.
\end{lemma}
	
\lng{\begin{proof}
	We set $w_2 = 1$ (this weight only plays a role later \textit{stateful} routing lower bound) and $w_1 < w_0 < (\ell \!- \! 1)w_1$ (more precise values are determined further below). Let $s_i, t_j$ be a source-target pair of $\Gamma$ where $e := \{s_i,t_j\} \in E$ is part of $G$ (but not necessarily $\Gamma$). For the case $x_e = 1$ we define $d_1 := d(s_i,t_j) = w_1\!+\!h$ and $d_0 := d(s_i,t_j) = w_0\!+\!h$ for the case $x_e = 0$ (c.f., Lemma \ref{lem:source_target_distance_weighted}).
	
	Let $\mathcal A$ be an algorithm solving the \textit{stateless} routing problem with the claimed approximation ratio (with const.\ probability). Let $P$ be the \textit{simple} $s_i$-$t_j$-path induced by the \textit{stateless} routing scheme computed by $\mathcal A$. Recall that $P$ must be simple as otherwise a packet that is oblivious to the prior routing decisions would be trapped in a loop. Our aim is that $s_i$ can decide whether $x_e =0$ or $x_e =1$ from the next routing node on $P$. That is, we want that $P$ contains $v$ if and only if $x_e = 0$. However, there are a few options to obtain the $s_i$-$t_j$-path $P$ that do, in certain cases, not abide by this requirement. We enumerate these in the following.
	
	Assume $P$ does not contain $e$. Then option (1) is to go the left ``lane'' via $v$ and then use the direct blue edge $\{v', t_j\}$ to get to $t_j$, i.e., the path of length $d_0$. If $P$ goes from $s_i$ directly to $u_i$, then $P$ can be completed into a path to $t_j$ by (2) using only edges that are also part of $G$ (but not $e$), i.e., only red edges in Figure \ref{fig:lower_bound_enhanced}. Note that option (2) must include at least $\ell \! - \! 1$ red edges due to the girth $\ell$ of $G$. Option (3) is where $P$ goes to $u_i$ first uses any red edge $\{u_i,t_p\}$ ($p \neq j$) and then the two blue edges $\{t_p, v'\}, \{v', t_j\}$ to reach $t_j$. Note that for the case $x_e = 0$ all other $s_i$-$t_j$-paths are either strictly longer than option (1),(2),(3) or not simple.
	The distances of the three paths (1),(2),(3) are at least (recall $w_2 = 1$):	
	\begin{align*}
	d^{(1)} & := w_0 + h \;(= d_0)\\
	d^{(2)} & := (\ell \! - \! 1)w_1 + h\\
	d^{(3)} & := w_1 + 2w_0 + h
	\end{align*}
	To enforce $v \notin P$ in case $x_e = 1$, we require that path (1) is unfeasible, i.e., exceeds the allowed distance $\alpha_\ell d_1$. Furthermore, to enforce $v \in P$ in case $x_e = 0$, we require that the paths (2) and (3) exceed the allowed distance $\alpha_\ell d_0$. From this we obtain the following conditions.
	\begin{align*}
	\alpha_{\ell} \cdot d_1 & \stackrel{!}{<} d^{(1)} = w_0 + h\tag{1}\\
	\alpha_{\ell} \cdot d_0 & \stackrel{!}{<} d^{(2)} = (\ell \! - \! 1)w_1 + h \tag{2} \\
	\alpha_{\ell} \cdot d_0 & \stackrel{!}{<} d^{(3)} = w_1 + 2w_0 + h \tag{3}
	\end{align*}
	The remaining part of the proof is merely technical, we need to maximize $\alpha_\ell$ under the above constraints. (Afterwards, the rest follows from having solved the node communication problem as in the proof of Lemma \ref{lem:lower_bound_distance_oracle}). Set $w_1 := t < w_0$ for some yet unspecified variable $t > 0$. Further, we set $\alpha := \frac{w_0}{t} - \eps$ and show that this fulfills Equation (1):
	\[
	\alpha_{\ell} \cdot d_1 = \big(\tfrac{w_0}{t} - \eps\big)(t+h) = w_0 + \tfrac{h w_0}{t} - \eps(t+h) < w_0 + h.
	\]
	For Equation (2) we obtain:
	\begin{align*}
	& \alpha_{\ell} \cdot d_0 = \big(\tfrac{w_0}{t} \!-\! \eps\big)(w_0\!+\!h)  < (\ell \! - \! 1)t + h\\
	\Longleftrightarrow \quad & \tfrac{w_0^2}{t} + \tfrac{w_0}{t}(h-\eps t)-\eps h - h < (\ell \! - \! 1)t\\
	\Longleftrightarrow \quad & w_0^2 + \underbrace{w_0(h-\eps t)}_{<0, \text{ for } t > h/\eps} \underbrace{-\eps ht - ht}_{<0} < (\ell \! - \! 1)t^2\\
	{\Longleftarrow} \quad & w_0^2 \leq (\ell \! - \! 1)t^2\\
	{\Longleftrightarrow} \quad & w_0 \leq t \cdot \sqrt{\ell \! - \! 1}.
	\end{align*}
	Now let us turn to Equation (3)
	\begin{align*}
	& \alpha_{\ell} \cdot d_0 = \big(\tfrac{w_0}{t} \!-\! \eps\big)(w_0\!+\!h)  < t + 2w_0 + h\\
	\Longleftrightarrow \quad & \tfrac{w_0^2}{t} + \tfrac{w_0}{t}(h-\eps t)-\eps h < t + 2w_0 + h\\
	\Longleftrightarrow \quad & w_0^2 + w_0(\underbrace{h-\eps t}_{<0}-2t) \underbrace{-\eps ht - ht}_{<0} - t^2 < 0\\
	\Longleftarrow \quad & w_0^2 -2tw_0 - t^2\leq 0 \\
	\Longleftarrow \quad & w_0 \leq t \cdot (1\!+\!\sqrt 2).
	\end{align*}
	Therefore, Equations (2),(3) are fulfilled if $w_0 \leq t \cdot \sqrt{\ell \! - \! 1}$ and $w_0 \leq t \cdot (1\!+\!\sqrt 2)$ (and \smash{$t > \frac{h}{\eps}$}, which can be chosen freely). This implies that the stretch $\alpha_\ell = \frac{w_0}{t} - \eps$ must satisfy $\alpha_\ell\leq \sqrt{\ell \! - \! 1} - \eps$ and $\alpha_\ell \leq 1\!+\!\sqrt 2 \!-\! \eps$, whereas the former condition for $\alpha_\ell$ dominates the latter if and only if $\ell \leq 8$.
	
	With this choice of $w_0, \alpha_\ell$, the first node that a packet from $s_i$ with destination $t_j$ is routed to is the node $v$ if and only if $x_e = 0$. Hence, the nodes in $B$ collectively learn $X$ from the information provided by algorithm $\calA$ and the labels of the nodes $t_j$, which therefore solves the node communication problem. The runtime and size of the labels then follows the same way as in the proof of Lemma \ref{lem:lower_bound_distance_oracle}.
\end{proof}	}

We plug graphs $G \in \calG_{k,\ell}$ into Lemma \ref{lem:lower_bound_stateless_routing}. Since in this case we get no improvements in the stretch for girth $\ell\geq 8$ it suffices to apply Lemma \ref{lem:low_even_girth_graph_density}. Beside the changed values for the stretch, the proof is the same as that of Theorem \ref{thm:lower_bound_distance_oracle_small_stretch}, we just have to use the corresponding values of $\delta$ from Lemma \ref{lem:low_even_girth_graph_density} for $\ell = 4,6,8$.

\begin{theorem}
	\label{thm:lower_bound_stateless_routing}
	Any algorithm that solves the \emph{stateless routing problem} in the \hybridpar{\infty}{\gamma} model with constant success probability with
	\begin{itemize}
		\item stretch $\sqrt{3}\!-\!\eps$ takes \smash{$\Omega\big((\tfrac{n}{\gamma})^{\frac{1}{3}}\big)$} rounds for label size \smash{$\leq c \cdot (n^2 \gamma)^{\frac{1}{3}}$}
		\item stretch $\sqrt{5}\!-\!\eps$ takes \smash{$\Omega\big(\tfrac{n^{1/5}}{\gamma^{2/5}}\big)$} 
		rounds for label size \smash{$\leq c \cdot (n^2 \gamma)^{\frac{1}{5}}$}
		\item stretch $1\!+\!\sqrt{2}\!-\!\eps$ takes
		\smash{$\Omega\big(\tfrac{n^{1/7}}{\gamma^{3/7}}\big)$}
		rounds for label size \smash{$\leq c \cdot (n^2 \gamma)^{\frac{1}{7}}$}
	\end{itemize}
	for any const.\ $\eps > 0$ and a fixed const.\ $c > 0$.
\end{theorem}

\subsection{Stateful Routing Scheme}

We obtain similar lower bound results for the approximate \textit{stateful} routing problem, however with even smaller stretch. Recall that in the stateful version the problem is relaxed in the sense that a routing decision may also depend on the information a packet has gathered along the previous routing path.

Since this permits loops in the routing path, it opens up additional options for routing a packet from $s_i$ to $t_j$ that we need to prohibit. For instance, a packet could first travel to $u_i$, then check if the direct edge $\{u_i,t_j\}$ is present, and if not travel back to $s_i$ to take the shorter route via $v$ instead. Note that this path has the same number of red and blue edges as the shortest path directly to $v$ and then to $t_j$ (c.f. Figure \ref{fig:lower_bound_enhanced}).

The trick is to make the weight $w_2$ (orange edges) of all incident edges of $s_i$ more expensive, such that revisiting $s_i$ breaks the approximation guarantee. This again forces the source $s_i$ to make the correct decision with the first node it routes the packet to, which renders the ability to travel in loops and learn along the way useless. \shrt{The proof of the following lemma is in Appendix \ref{sec:lower_bounds_approx_proofs}.}


\begin{lemma}
	\label{lem:lower_bound_stateful_routing}
	Consider $\Gamma$ constructed from $G = (A,E) \in \mathcal G_{k,\ell}$ with $|E| = \Theta\big(k^{1+\delta}\big)$ edges for some $\delta > 0$. For any constant $\varepsilon > 0$ let \smash{$\alpha_4 = \sqrt 2 - \eps$}, \smash{$\alpha_6 = \tfrac{5}{3} - \eps$}, \smash{$\alpha_8 = \tfrac{7}{4} - \eps$}. For $\ell \geq 10$ let \smash{$\alpha_\ell = \frac{3+\sqrt{17}}{4} - \eps \approx 1.78$}.
	Any algorithm that computes a \emph{stateful routing scheme} on $\Gamma$ with stretch $\alpha_\ell$ and constant success probability in the \hybridpar{\infty}{\gamma} model takes \smash{$\Omega\big((\tfrac{n^\delta}{\gamma})^{\frac{1}{2+\delta}}\big)$} rounds. This holds for labels of size 
	\smash{$c \!\cdot\! n^{\frac{2\delta}{2+\delta}} \!\cdot\! \gamma^{\frac{\delta}{2+\delta}}$} and fixed constant $c>0$.
\end{lemma}

\lng{\begin{proof}
		
		The beginning of the proof is similar to the one of Lemma \ref{lem:lower_bound_stateless_routing}.	
		Consider algorithm $\mathcal A$ that solves the \textit{stateful} routing problem with the claimed approximation ratio and constant probability. Let $s_i, t_j \in V_\Gamma$ with $e := \{s_i,t_j\} \in E$ and define $d_1 := d(s_i,t_j) = w_1\!+\!w_2\!+\!h\!-\!1$ and $d_0 := d(s_i,t_j) = w_0\!+\!w_2\!+\!h\!-\!1$ for the cases $x_e = 1$, $x_e = 0$ respectively (we will ensure $w_1 < w_0 < (\ell\!-\!1)w_1$ so that Lemma \ref{lem:source_target_distance_weighted} applies). 
		
		Let $P$ be the $s_i$-$t_j$-path induced by the \textit{stateful} routing scheme computed by $\mathcal A$. Similar to before, our goal is to show that the first node on $P$ is $v$ if and only if $x_e = 0$, so that $s_i$ learns $x_e$ from its routing decision. As before, we aim to prohibit all routing options (i.e., make them break the stretch guarantee) that do not abide by this requirement. We will extend our list of options for $P$ which are in some cases undesirable from the previous proof.
		
		We have the loop-less routing options (1),(2),(3) from before (c.f., proof of Lemma \ref{lem:lower_bound_stateless_routing}), whose lengths change by an additive term $w_2-1$ due to the introduction of weight $w_2$ (orange edges in Figure \ref{fig:lower_bound_enhanced}). Note that any path (possibly with loops) that contains $s_i$ just once is at least as long as one of the options (1)-(3) (in their respective cases) so prohibiting the latter prohibits the former.
		
		
		The problem that arises is from visiting $s_i$ at least twice is that it can mislead $s_i$ by first going to $v$ even though $x_e = 1$ or vice versa. 
		Assume the case that $e$ is not in $\Gamma$ ($x_e = 0$), then we need to prohibit routing option (4) that visits the first node on the path towards $u_i$, returns to $s_i$, travels directly to $v'$ and uses the blue edge to $t_j$.		
		Conversely, assuming the case that $e$ is present in $\Gamma$ ($x_e = 1$), $P$ we need to prohibit option (5), which visits $v$ first but then returns to $s_i$, travels to $u_i$ and uses $e$ to reach $t_j$. Note that all paths that contain $s_i$ at least twice, are at least as long as one of the paths (4),(5) (in their respective cases). The described routing paths (1)-(5) have the following respective lengths:
		\begin{align*}
		d^{(1)} & := w_0 + w_2 + h -1\\
		d^{(2)} & := (\ell \! - \! 1)w_1 + w_2 + h -1\\
		d^{(3)} & := 2w_0 + w_1 + w_2 + h -1\\
		d^{(4)} & := w_0 + 3w_2 + h -1\\
		d^{(5)} & := w_1 + 3w_2 + h -1
		\end{align*}	
		
		We obtain inequalities similar as before. For the reasoning of inequalities (1)-(3) consider the proof of Lemma \ref{lem:lower_bound_stateless_routing}. Observe that if $P$ would follow path options (4),(5) it could make $s_i$ believe $x_e = 1$ or $x_e = 0$ respectively, even though the opposite is true. We require that path options (4) and (5) exceed the stretch $\alpha_\ell$ in the cases in which they are undesirable, that is, $x_e = 0$ and $x_e = 1$, respectively. 
		\begin{align*}
		\alpha_{\ell} \cdot d_1 & \stackrel{!}{<} d^{(1)} = w_0 + w_2 + h -1\tag{1}\\
		\alpha_{\ell} \cdot d_0 & \stackrel{!}{<} d^{(2)} = (\ell \! - \! 1)w_1 + w_2 + h -1 \tag{2} \\
		\alpha_{\ell} \cdot d_0 & \stackrel{!}{<} d^{(3)} = 2w_0 + w_1 + w_2 + h -1 \tag{3} \\
		\alpha_{\ell} \cdot d_0 & \stackrel{!}{<} d^{(4)} = w_0 + 3w_2 + h -1 \tag{4} \\
		\alpha_{\ell} \cdot d_1 & \stackrel{!}{<} d^{(5)} = w_1 + 3w_2 + h -1 \tag{5}
		\end{align*}
		This leaves us with an optimization problem where $\alpha_{\ell}$ is to be maximized in the domain $1 \leq w_1 \leq w_0,w_2$, subject to conditions (1)-(5). This is admittedly a bit tedious, in particular since the optimization problem has a different result for girth $\ell \in \{4,6,8,10\}$ (we gain no improvement for larger $\ell$). We do not reproduce all necessary calculations of the optimization in detail, instead we give some explanations to make our results (given in Table \ref{tbl:lower_bounds_stateful_opt_values}) reproducible. 
		
		One obstacle is the presence of edges with weight 1 in $\Gamma$ forming the path of length $h-1$, which is not yet fixed and therefore prohibits solving the optimization problem directly. Our strategy is to first relax our problem to a simplified one with zero-weight edges. That is, we set all edges, except those with weight $w_0,w_1,w_2$, to zero and adapt the distances $d_0, d_1, d^{(1)}, \dots,  d^{(5)}$ accordingly (essentially slashing $h\!-\!1$ in those distances). We also set $w_1 = 1$. This eliminates $w_1,h$ from conditions (1)-(5) and the resulting maximization problem(s) can be solved directly.
		
		From the results of the simplified optimization problems (for girths $\ell \in \{4,6,8,10\}$), we obtain almost optimal solutions for the general optimization problem with non-zero weights as follows. We multiply the values of $w_0,w_1,w_2$ with a (sufficiently large) parameter $t \geq 1$ and subtract an $\eps>0$ from $\alpha_\ell$ to accommodate a (small) slack that is required in the general inequalities. The value of $t$ depends on the slack $\eps$ in $\alpha_\ell$ that is given but is generally polynomial in $n$ as long as $\eps$ is constant.
		The results obtained using this procedure are given in Table \ref{tbl:lower_bounds_stateful_opt_values}.
		
		\begin{table}
			\centering
			\begin{tabular}{p{5mm}>{\centering\arraybackslash}p{14mm}>{\centering\arraybackslash}p{14mm}>{\centering\arraybackslash}p{6mm}>{\centering\arraybackslash}p{14mm}}
				\toprule
				$\ell$ & $\alpha_\ell$ & $w_0$ & $w_1$ & $w_2$\\
				\midrule
				4 & $\sqrt 2 - \eps$ & $2t\sqrt 2 \!-\!t$ & $t$ & $t$\\[1.5mm]
				6 & $\frac{5}{3} - \eps$ & $\frac{5t}{2}$ & $t$ & $\frac{5t}{4}$ \\[1.5mm]
				8 & $\tfrac{7}{4} - \eps$ & $\frac{35t}{11}$ & $t$ & $\frac{21t}{11}$ \\[1.5mm]
				10 & $\frac{3+\sqrt{17}}{4} - \eps$ & $\frac{3+\sqrt{17}}{2}t$ & $t$ & $\frac{5+\sqrt{17}}{4}t$\\
				\bottomrule
			\end{tabular}
			\caption{\lng{\boldmath} Results of maximizing $\alpha_\ell$, s.t., conditions (1)-(5) for $\ell \in \{4,6,8,10\}$ (and scaling with $t$ and introducing slack $\eps$).}
			\label{tbl:lower_bounds_stateful_opt_values}
		\end{table}
		
		The most relevant parameter in Table \ref{tbl:lower_bounds_stateful_opt_values} are certainly the values for $\alpha_\ell$ for $\ell \in \{4,6,8,10\}$, whereas we also provide the weights $w_0,w_1,w_2$ for reproducibility. Showing that the given parameters do in fact satisfy conditions (1)-(5) for any constant $\eps > 0$ and some choice of $t$, is a repetitive task. We show this once for $\ell=10$ and condition (1) (which is almost tight in this case), the other cases can be repeated analogously.
		\begin{align*}
		\alpha_{\ell} \cdot d_1 &  = \big(\tfrac{3+\sqrt{17}}{4} \big)d_1 - \eps d_1\\
		& = \big(\tfrac{3+\sqrt{17}}{4} \big)(w_1 + w_2 + h -1) - \eps d_1\\
		& = \big(\tfrac{3+\sqrt{17}}{4} \big)(w_1 + w_2) + h \!-\!1 + \smash{\underbrace{\big(\tfrac{\sqrt{17}-1}{4} \big)(h\!-\!1) - \eps d_1 }_{<0,\text{ for large }t}}\\
		& < \big(\tfrac{3+\sqrt{17}}{4} \big)\big(\tfrac{9+\sqrt{17}}{4}\cdot t\big) + h \!-\!1\\
		& = \big(\tfrac{11+3\sqrt{17}}{4}\big) \cdot t + h \!-\!1\\
		& = \big(\tfrac{3+\sqrt{17}}{2} + \tfrac{5+\sqrt{17}}{4}\big) \cdot t + h \!-\!1\\
		& = w_0 + w_2 + h \!-\!1 = \smash{d^{(1)}}.
		\end{align*}
		Note that, since $d_1$ grows linear in $t$ the inequality in the fourth line holds for \smash{$t > \big(\tfrac{\sqrt{17}-1}{4\eps} \big)(h\!-\!1) \in \bigO(n)$}.
		
		Finally, by enforcing conditions (1)-(5) with the appropriate parameters given in Table \ref{tbl:lower_bounds_stateful_opt_values}, we can guarantee that the first node that a packet from $s_i$ to $t_j$ is first routed to $v$ if and only if $x_e = 0$. From this the nodes in $B$ collectively learn $X = (x_e)_{e \in E}$ thus solving the node communication problem. The runtime and size of the labels follows as in the proof of Lemma \ref{lem:lower_bound_distance_oracle}.
\end{proof}}
	
	Again, our actual lower bounds come from inserting graphs $G \in \calG_{k,\ell}$ into Lemma \ref{lem:lower_bound_stateless_routing}. Our best stretch is obtained for $\ell = 10$, but unfortunately we have a gap for that value in Lemma \ref{lem:low_even_girth_graph_density}. Therefore, for the largest stretch value we use a graph $G \in \calG_{k,12} \subseteq \calG_{k,10}$, which has the drawback of not being as dense. Aside from different stretch values, the proof follows that of Theorem \ref{thm:lower_bound_distance_oracle_small_stretch}, by inserting the values of $\delta$ from Lemma \ref{lem:low_even_girth_graph_density} for $\ell = 4,6,8,12$.
	
	\begin{theorem}
		\label{thm:lower_bound_stateful_routing}
		Any algorithm that solves the \emph{stateful routing problem} in the \hybridpar{\infty}{\gamma} model with constant success probability with
		\begin{itemize}
			\item stretch $\sqrt{2}\!-\!\eps$ takes \smash{$\Omega\big((\tfrac{n}{\gamma})^{\frac{1}{3}}\big)$} rounds for label size \smash{$\leq c \cdot (n^2 \gamma)^{\frac{1}{3}}$}
			\item stretch $\frac{5}{3}\!-\!\eps$ takes \smash{$\Omega\big(\tfrac{n^{1/5}}{\gamma^{2/5}}\big)$} 
			rounds for label size \smash{$\leq c \cdot (n^2 \gamma)^{\frac{1}{5}}$}
			\item stretch $\frac{7}{4}\!-\!\eps$ takes
			\smash{$\Omega\big(\tfrac{n^{1/7}}{\gamma^{3/7}}\big)$}
			rounds for label size \smash{$\leq c \cdot (n^2 \gamma)^{\frac{1}{7}}$}
			\item stretch \smash{$\frac{3+\sqrt{17}}{4}\!-\!\eps$} takes
			\smash{$\Omega\big(\tfrac{n^{1/11}}{\gamma^{5/11}}\big)$}
			rounds for label size \smash{$\leq c \cdot (n^2 \gamma)^{\frac{1}{11}}$}
		\end{itemize}
		for any const.\ $\eps > 0$ and a fixed const.\ $c > 0$.
	\end{theorem}

	\shrt{\newpage}

\appendix

\section{Some Basic Probabilistic Concepts}
\label{apx:generalnotations}

\begin{lemma}[Chernoff Bound]
	\label{lem:chernoffbound}
	We use the following forms of Chernoff bounds in our proofs:
	$$\mathbb{P}\big(X > (1 \!+\! \delta) \mu_H\big) \leq \exp\Big(\!-\!\frac{\delta\mu_H}{3}\Big),$$
	with $X = \sum_{i=1}^n X_i$ for i.i.d.\ random variables $X_i \in \{0,1\}$ and $\mathbb{E}(X) \leq \mu_H$ and $\delta \geq 1$. Similarly, for $\mathbb{E}(X) \geq \mu_L$ and $0 \leq \delta \leq 1$ we have
	$$\mathbb{P}\big(X < (1 \!-\! \delta) \mu_L\big) \leq \exp\Big(\!-\!\frac{\delta^2\mu_L}{2}\Big).$$
\end{lemma}

\begin{remark}
	\label{rem:chernoffbound}
	Note that the first inequality even holds if we have $k$-wise independence among the random variables $X_i$ for $k \geq \lceil \mu_H \delta \rceil$ (c.f., \cite{Schmidt1995} Theorem 2, note that a substitution $\mu_H := (1 \!+\! \eps)\E(X)$ generalizes the result for any $\mu_H \geq \E(X)$).
\end{remark}

\begin{lemma}[Union Bound]
	\label{lem:unionbound}
	Let $E_1, \ldots ,E_k$ be events, each taking place w.h.p. If $k \leq p(n)$ for a polynomial $p$, then $E \coloneqq \bigcap_{i=1}^{k} E_i$ also takes place w.h.p.
\end{lemma}

\begin{proof}
	Let $d \coloneqq \deg(p)\!+\!1$. Then there is an $n_0 \geq 0$ such that $p(n) \leq n^d$ for all $n \geq n_0$. Let $n_1, \ldots , n_k \in \mathbb{N}$ such that for all $i \in \{1, \ldots, k\}$ we have $\mathbb{P}(\overline{E_i}) \leq \tfrac{1}{n^c}$ for some (yet unspecified) $c > 0$.
	With Boole's Inequality (union bound):
	\begin{align*}
	\mathbb{P}\big(\overline{E}\big) \!= \mathbb{P}\Big(\bigcup_{i=1}^{k} \overline{E_i} \Big) \leq \sum_{i=1}^{k} \mathbb{P}(\overline{E_i}) \leq \sum_{i=1}^{k} \!\frac{1}{n^c} \leq \frac{p(n)}{n^{c}} \leq \frac{1}{n^{c-d}}
	\end{align*}
	for all $n \geq n_0' \coloneqq \max(n_0, \ldots ,n_k)$. Let $c' > 0$ be arbitrary. We choose $c \geq c' \!\!+\! d$. Then $\mathbb{P}\big(\overline{E}\big) \leq \frac{1}{n^{c'}}$ for all $n \geq n_0'$.
\end{proof}

\begin{remark}
	If a constant number of events is involved we use the above lemma without explicitly mentioning it. It is possible to use the lemma in a nested fashion as long as the number of (nested) applications is polynomial in $n$.
\end{remark}

Our main application for the above two lemmas is the proof of Lemma \ref{lem:node_sampling}.

\begin{proof}[Proof of Lemma \ref{lem:node_sampling}]
	Let $u,v \in V$ with $hop(u,v) \!\geq\! \xi x \ln n$. Fix a shortest $u$-$v$-path $P_{u,v}$ and let $Q$ be a sub-path of $P_{u,v}$ with at least $\xi x \ln n$ nodes. Let $X_{u,v}$ be the random number of marked nodes on $Q$. Then we have \smash{$\mathbb{E}(X_{u,v}) \geq \frac{|Q|}{x} \geq {\xi \ln n}$}. Let $c > 0$ be arbitrary. We use a Chernoff bound (c.f.\ Lemma \ref{lem:chernoffbound}):
	\[\mathbb{P}\Big( X_{u,v} < \frac{\xi \ln n}{2}\Big) \leq \exp\Big(\!\!-\! \frac{\xi \ln n}{8}\Big) \stackrel{\xi \geq 8c}{\leq} \frac{1}{n^c}.\]
	Thus we have $X_{u,v} \geq 1$ w.h.p.\ for constant \smash{$\xi \geq \max(8c, 2 / \ln n)$}. Therefore the claim holds w.h.p.\ for the pair $u,v$. We claim that w.h.p.\ the event $X_{u,v} \geq 1$ occurs for all pairs $u,v \in V$ and for all sub-paths $Q$ of $P_{u,v}$ longer than $\xi x \ln n$ hops, for at least one shortest path $P_{u,v}$ from $u$ to $v$. There are at most \smash{$n^2$} many pairs $u,v \in V$. Moreover we can select at most $n$ sub-paths $Q$ of $P$ that do not fully contain any other selected sub-path. Hence the claim follows with the union bound given in Lemma~\ref{lem:unionbound}.
\end{proof}

\section{Concepts from Information Theory}
\label{sec:information_theory}

The Entropy $H(X) \!:=\! - \!\sum_{x \in S} \mathbb{P}(X \!=\! x) \log \big(\mathbb{P}(X \!=\! x) \big)$ gives a lower bound for expected number of bits required for encoding the state of a random variable. This is entailed by Shannon's \cite{Shannon1948} source coding theorem.

\begin{lemma}[c.f., \cite{Shannon1948}]
	\label{lem:source_coding_theorem}
	Given a random variable $X$ with outcomes from some set $S$ and an uniquely decodable code $f : S \to \{0,1\}^*$ with expected code length $\mathbb{E}(|f(X)|)$. Then $\mathbb{E}(|f(X)|) \geq H(X)$.
\end{lemma}

In particular, in a two party communication setting (see Definition \ref{def:party_comm_problem}) this implies that $H(X)$ constitutes a lower bound for the worst case number of bits that have to be transmitted from one party that knows the state of $X$ to some party that needs to learn it.

\begin{corollary}
	\label{cor:lower_bound_two_party}
	Bob must receive at least $H(X)$ bits from Alice in expectation, as part of any protocol solving the two party communication problem (Def.\ \ref{def:party_comm_problem}).
\end{corollary}

\begin{proof}
	Assume, for a contradiction, that we have a protocol $\mathcal P$ in which sending \textit{less} than $H(X)$ bits from Alice to Bob always suffices to solve the two party communication problem. 	
	Clearly, for any two possible outcomes $x_1, x_2 \in S$ of $X$, the transcript of the communication occurring between Alice and Bob must be different as otherwise Bob would not be able to distinguish $x_1$ from $x_2$.	
	But then we could use the transcript of $\mathcal P$ for any given outcome of $x \in S$ of $X$ as uniquely decodable code for $x$ of expected length less than $H(X)$, a contradiction to Lemma \ref{lem:source_coding_theorem}.
\end{proof}

Using information theoretic concepts, the above statement generalizes for a protocol that has a probability of at least $p$ that Bob can successfully decode the state of $X$ after it is terminates.

\begin{lemma}
	\label{lem:lower_bound_two_party}
	Bob must receive at least $p\cdot H(X)-1$ bits from Alice in expectation, as part of any protocol that solves the two party communication problem (c.f., Def.\ \ref{def:party_comm_problem}) with probability at least $p$.
\end{lemma}

\begin{proof}	
	We assume that the random variable $X$ has a finite number of outcomes (which is sufficient for our purposes), i.e., $X \in \{x_1, \dots ,x_k\}$ for some $k \in \mathbb N$. Assuming the outcome $X=x_i$, let $y_i$ be the output that Bob makes after the randomized communication protocol terminates. Then 
		\begin{align*}
			y_i =
			\begin{cases}
				x_i, & \text{with probability } p_i\\
				x_j \text{ and } j \neq i, & \text{with probability } (1-p_i),
			\end{cases}
		\end{align*}
	where $p \leq p_i \leq 1$. That means we have another random variable $Y$ dependent on $X$, which describes Bob's guess about the state of $X$. It remains to prove that the information about $X$ that is still contained in $Y$, is large. This is known as the \textit{transinformation} $I(X;Y)$ and since Bob ``learns'' the state of $Y$, at least $I(X;Y)$ must have been transmitted from Alice to Bob. In particular, we want to show $I(X;Y) \geq p \cdot H(X) - 1$.
	The transinformation is defined as 
	\[
		I(X;Y) = \sum_{i,j \in [k]} \mathbb{P}(X\!=\!x_i,Y\!=\!x_j) \cdot \log \frac{\mathbb{P}(X\!=\!x_i,Y\!=\!x_j)}{\mathbb{P}(X\!=\!x_i)\mathbb{P}(Y\!=\!x_j)}
	\]
	Analyzing this directly is tricky since the output distribution of $Y$ for the second case, where $Y \neq X$, is not specified (and can not be made such, without loosing the generality of the claim). So we have to take a detour by defining a third random variable $Z$ that tells us if the protocol was successful.
	\begin{align*}
		Z =
		\begin{cases}
		1, & \text{if } y_i = x_i\\
		0, & \text{else}.
		\end{cases}
	\end{align*}
	To simplify the analysis of the transinformation we assume that Bob gets to know $Z$ ``for free'' and since $H(Z) \leq 1$ the additional information about $X$ from learning $Y$ is not significantly reduced. Formally, we first show $I(X;Y) \geq I(X;Y,Z) - H(Z)$ which allows us to analyze $I(X;Y,Z)$ instead. 
	
	The conditional entropy $H(A|B)$ describes the amount of ``new'' information in some random variable $A$ given that we already know random variable $B$. 	
	In the following steps we will use the fact that $H(Z|X,Y) = 0$ since $Z$ is functionally dependent on $X$ and $Y$ and we will the chain rule of entropy $H(A,B) = H(A|B) \!+\! H(B)$. We plug this into the alternative characterization of transinformation
	\begin{align*}
	I(X;Y,Z) & = H(X) - H(X| Y,Z) \tag*{\small \text{\textit{def.\ of}} $I(X;Y,Z)$}\\
	& = H(X) - H(X,Z|Y) + H(Z|Y) \tag*{\small \text{\textit{chain rule}}}\\
	& = H(X) - H(Z|X,Y) - H(X|Y) + H(Z|Y) \tag*{\small \text{\textit{chain rule}}}\\
	& = H(X) - H(X|Y) + H(Z|Y) \tag*{\small $H(Z|X,Y) = 0$}\\
	& = I(X;Y) + H(Z|Y) \tag*{\small \text{\textit{def.\ of}} $I(X;Y)$}\\
	& \leq I(X;Y) + H(Z).
	\end{align*}
	This implies $I(X;Y) \geq I(X;Y,Z) - H(Z)$, and it remains to show that $I(X;Y,Z)$ is large. 	
	The random variable $Z$ helps in the following way.
	For any $x_i$ we have
	\[
		\mathbb P(Y \!\!=\!\! x_i,Z \!\!=\!\! 1) = p_i \cdot \mathbb P(X \!\!=\!\! x_i) = \mathbb P(X \!\!=\!\! x_i, Y \!\!=\!\! x_i,Z \!\!=\!\! 1),
	\] 
	since $Z=1$ means that $Y=x_i$ is only possible if $X=x_i$. We obtain
	\begin{align*}
		I(X;Y,Z) & = \hspace*{-6mm}\sum_{i,j \in [k], z \in \{0,1\}} \hspace*{-6mm} \mathbb P(X \!\!=\!\! x_i, Y \!\!=\!\! x_j,Z \!\!=\!\! z) \cdot \log \frac{\mathbb P(X \!\!=\!\! x_i, Y \!\!=\!\! x_j,Z \!\!=\!\! z)}{\mathbb P(X \!\!=\!\! x_i) \cdot \mathbb P(Y \!\!=\!\! x_j,Z \!\!=\!\! z)}\\
		& \geq \hspace*{-0mm}\sum_{i \in [k]} \hspace*{-0mm} \mathbb P(X \!\!=\!\! x_i, Y \!\!=\!\! x_i,Z \!\!=\!\! 1) \cdot \log \frac{\mathbb P(X \!\!=\!\! x_i, Y \!\!=\!\! x_i,Z \!\!=\!\! 1)}{\mathbb P(X \!\!=\!\! x_i) \cdot \mathbb P(Y \!\!=\!\! x_i,Z \!\!=\!\! 1)}\\
		& = \hspace*{-0mm}\sum_{i \in [k]} \hspace*{-0mm}  p_i \cdot \mathbb P(X \!\!=\!\! x_i) \cdot \log \frac{1}{\mathbb P(X \!\!=\!\! x_i)}\\
		& \geq \: p \cdot \hspace*{-0mm}\sum_{i \in [k]} \hspace*{-0mm}  \mathbb P(X \!\!=\!\! x_i) \cdot \log \frac{1}{\mathbb P(X \!\!=\!\! x_i)} = p \cdot H(X)
	\end{align*}	
	Finally, we have 
	$I(X;Y) \geq I(X;Y,Z) - H(Z) \geq p \cdot H(X) - 1.$
\end{proof}

\section{Density of Bounded Girth Graphs}
\label{sec:girth_density}

We reproduce a few known and conjectured results from extremal graph theory, in particular that the number of edges in cycle-free graphs can be bounded from above and below. We are going to formulate these results in the context and granularity that we require in this article (neglecting constants, in particular). 
First, there is a long standing conjecture from Erd\H{o}s and Simonovits \cite{Erdoes1982}.\footnote{\cite{Erdoes1982} states in Conjecture 5 that there are graphs without cycles of a fixed length with the claimed density, and conjectures that the same holds for excluding smaller cycles as well (below Theorem 2 of \cite{Erdoes1982}).}

\begin{conjecture}[by \cite{Erdoes1982}]
	\label{con:erdoes_girth}
	For any $k \!\in\! \mathbb{N}$, there is an $n$-node graph with girth $\geq 2k\!+\!1$ and $\Theta(n^{1+\frac{1}{k}})$ edges.
\end{conjecture}

It is known that a graph with average degree $d$ and girth $2k+1$ has $n \in \Omega(d^k)$ nodes due to \cite{Alon2001}. This translates into the following lemma:

\begin{lemma}[c.f., \cite{Alon2001}]
	\label{lem:girth_edges_bound}
	Any $n$-node graph with girth at least $2k+1, k \in \mathbb N$ has at most \smash{$O(n^{1+\frac{1}{k}})$} edges.
\end{lemma}

Conjecture \ref{con:erdoes_girth} is known to be true for some parameters of $k$ due to \cite{Singleton1966} and \cite{Benson1966}.

\begin{lemma}[c.f., \cite{Singleton1966}, \cite{Benson1966}]
	\label{lem:low_girth_graph_density}
	For $k=2,3,5$ there are $n$-node graphs with girth $2k\!+\!1$ and \smash{$\Theta(n^{1+\frac{1}{k}})$} edges.
\end{lemma}

There are more general lower bounds for graphs for arbitrary girth by \cite{Lazebnik1997} which the survey \cite{Verstraete2016} summarizes as follows:

\begin{lemma}[c.f., \cite{Lazebnik1997}, \cite{Verstraete2016}]
	\label{lem:high_girth_graph_density}
	For any $k\geq 2$ there is a $n$-node graph with girth $2k\!+\!1$ and \smash{$\Theta(n^{1+\frac{2}{3k-2}})$} edges if $k$ is even, and \smash{$\Theta(n^{1+\frac{2}{3k-3}})$} if $k$ is odd.
\end{lemma}

Above we mention only uneven girth, whereas in this paper we are mostly interested in (balanced) bipartite graphs which naturally have even girth. Note that given a graph with girth $2k+1$, one easily obtains a balanced, bipartite graph of even girth $2k+2$ with the same asymptotic order and size by constructing the bipartite double cover.

\begin{lemma}
	\label{lem:make_girth_even}
	Let $G\!=\!(V,E)$ be a $n$-node graph with girth $2k\!+\!1$, then there is a balanced, bipartite graph $G'\!=\!(V',E')$ with girth $2k\!+\!2$, $|V'| \!=\! 2|V|$ and $|E'| \!=\! 2|E|$.
\end{lemma}

\begin{proof}
	Let \smash{$V' := \bigcup_{v \in V} \{v_1, v_2\}$}, i.e., for each node $v \in V$ we create two copies. Further, let \smash{$E' = \bigcup_{\{u,v\} \in E} \{\{u_1,v_2\},\{v_1,u_2\}\}$}, i.e., for each edge $\{u,v\}$ in $E$ we create two ``crossing'' edges between the node copies $u_1,v_2$ and $u_2,v_1$. Any cycle of $G'$ must form a corresponding cycle in $G$, by taking the original edge $\{u,v\}$ for each edge $\{u_1,v_2\}$ in that cycle. Thus $G'$ can not have a cycle shorter than $2k+1$. Further, by construction, we have a (balanced) bipartition of $G'$ given by the nodes with index $1$ and $2$, respectively. Since $G'$ is bipartite, it can not contain an odd cycle, hence the girth is at least $2k+2$.
\end{proof}

Combining Lemma \ref{lem:make_girth_even} with Lemma \ref{lem:low_girth_graph_density} and the $n$-node clique which has girth 3 and $\Theta(n^2)$ edges, we obtain the following lemma.

\begin{lemma}
	\label{lem:low_even_girth_graph_density}
	For $\ell = 4,6,8,12$ there are balanced, bipartite $n$-node graphs with girth $\ell$ and \smash{$\Theta(n^{1+\frac{2}{\ell-2}})$} edges.
\end{lemma}

Note that Lemma \ref{lem:low_even_girth_graph_density} this is tight, since for any even $\ell\geq 4$ we obtain the upper bounds \smash{$\Omega(n^{1+\frac{2}{\ell-2}})$} in Lemma \ref{lem:girth_edges_bound} by plugging in the smaller uneven girth $\ell-1$.
For all other even girths we have to fall back on Lemma \ref{lem:high_girth_graph_density}. Combining it with Lemma \ref{lem:make_girth_even} gives us the lemma below. Note that we do not apply this lemma for girth $10$ as we can get the same asymptotic number of edges for the higher (= better) girth $12$ from Lemma \ref{lem:low_even_girth_graph_density}.

\begin{lemma}
	\label{lem:high_even_girth_graph_density}
	For any even $\ell \geq 14$ there is a balanced, bipartite $n$-node graph with girth $\ell$ and \smash{$\Theta(n^{1+\frac{4}{3\ell-10}})$} edges if $\ell \equiv 2 \!\!\mod 4$, or \smash{$\Theta(n^{1+\frac{4}{3\ell-12}})$} edges if $\ell \equiv 0 \!\!\mod 4$.
\end{lemma}

\section{Lower Bounds for LOCAL and NCC}

To quantify the advantage of combining two communication modes in the \hybrid model, we are interested in the communication complexity of computing routing schemes and distance oracles of the constituent communication modes \LOCAL and \NCC. The according lower bounds are significantly higher than the upper bounds for the \hybrid model given in Section \ref{sec:upper_bounds}.
Since we do not consider this part as the main scope of this paper we restrict ourselves to proof sketches, which can be completed into full proofs with moderate effort.

We start with the \LOCAL model (corresponds to \hybridpar{\infty}{0}, c.f., Definition \ref{def:hybrid_model}), where practically all shortest path problems are acknowledged to be of global nature, i.e., they require at least $\Omega(D_G)$ rounds of communication, in general (where $D_G$ is the hop diameter of the graph $G$ and $D_G \in \Omega(n)$). Since all graph problems can be trivially solved in $\bigO(D_G)$ rounds in \LOCAL (by collecting the whole graph at one or all nodes), global problems are usually uninteresting in \LOCAL (unless supplemented with additional capabilities, like in \hybrid or constrained further, like in \CONGEST).
We show that the same is also true for computing routing schemes and distance oracles.

\begin{lemma}
	\label{lem:lower_bound_local}
	Computing distance oracles and \textit{stateless} routing schemes with constant stretch and exact \textit{stateful} routing schemes in the \LOCAL model with success probability strictly larger than $\frac{1}{2}$ takes $\Omega(n)$ rounds irrespective of the allowed label size.
\end{lemma}

\begin{proof}[Proof Sketch]
	We start with routing schemes and subsequently show how the argument can be adapted to distance oracles. Consider a path with $n\!-\!1$ nodes with randomly assigned identifiers in $[n\!-\!1]$, such that each permutation of $[n\!-\!1]$ along he path has the same probability. We pick a random node $v$ i.i.d., from the middle part of that path, i.e., from the nodes with distance at least $\ell$ from either end of the path for some $\ell \in \Omega(n)$ with $\ell < \frac{n}{2}$ ($v$ is unaware of $\ell$). Call one side of the path ``left'' the other ``right'' of $v$ (however, nodes have no concept of left or right).
	
	A dedicated node $u$ with $ID(u)=n$ is attached either to the left or right end of the path each with probability $\frac{1}{2}$.
	If $v$ sends the packet in the opposite direction of $u$, the \textit{exact} routing scheme immediately fails and having a stateful routing scheme does not help for the first routing decision. If $v$ does not know in which direction $u$ is, then sending the packet in an arbitrary direction fails with probability at least $\frac{1}{2}$.
	
	Presume we have $k < \ell$ rounds to compute the routing scheme that succeeds with probability higher than $\frac{1}{2}$. Note that this does suffice that $v$ sees one end of the path. Thus $v$ must learn which direction $u$ is from combining its local information with the label $\lambda(u)$ that $u$ computed. However, the information nodes $u$ and $v$ can gather in $k$ rounds in the \LOCAL model is restricted to their respective $k$-hop neighborhoods $\calB_{u,k},\calB_{v,k}$ of $u$ and $v$. Thus the most information about $G$ that $\lambda(u)$ can contain is $\calB_{u,k}$, so let us assume $\lambda(u) = \calB_{u,k}$. 
	
	If $\calB_{u,k} \cap \calB_{v,k} = \emptyset$, then node $v$ can not distinguish scenarios where $\calB_{u,k}$ is ``left'' from those where $\calB_{u,k}$ is to the ``right'' of $v$. Furthermore, due to the symmetry of our random setup, both of these possibilities have exactly the same probability. Thus $v$ can not make a decision that succeeds with probability better than $\frac{1}{2}$. The only way that $v$ can learn the direction of $u$ is if $\calB_{u,k} \cap \calB_{v,k} \neq \emptyset$. But since $\hop(u,v) \geq \ell \in \Omega(n)$ we have $k \in \Omega(n)$.
	
	This argument can also be applied to distance oracles and stateless routing schemes. Since here the first decision of $v$ is final (in stateless routing a packet may not reverse direction on a path) we can also accommodate any constant stretch. We randomly pick either the left side or the right side with probability $\frac{1}{2}$ and make it a constant factor longer than the other, so $v$ must know the location of $u$ to be able to give an approximation that is better than this constant. (Note that with a slightly more involved argument we can also accommodate a small fixed constant stretch for stateful routing schemes).
\end{proof}

Admittedly, it seems artificial to consider our problems only in the \NCC model (\hybridpar{0}{\polylog n}, c.f., Definition \ref{def:hybrid_model}), because computing distance oracles or routing schemes for a local network suggests that this infrastructure could and should also be used for that. However, for the sake of theoretical comparison of the models, we assume that each node knows its neighbors in some local communication graph $G$, but can only communicate within the restrictions of the \NCC model.

\begin{lemma}
	\label{lem:lower_bound_ncc}
	Computing exact distance oracles and stateful routing schemes in the \NCC model with constant success probability takes $\tilOm(n)$ rounds even for labels of size $cn$ for some constant $c > 0$.
\end{lemma}

\begin{proof}[Proof Sketch]
	We sketch this proof based on our arguments for the lower bounds in the \hybrid model, so knowledge thereof is required. Although this proof sketch is rather informal, the rigorous formal arguments can be derived from (a much simplified form of) our lower bounds for the \hybrid model with moderate effort. 	
	We will use the basic setup of the unweighted graph construction in Definition \ref{def:unweighted_construction} depicted in Figure \ref{fig:lower_bound_basic}. Since there is no need to prohibit local communication with a long path between node sets $A$ and $B$, we are allowed to set $h = 1$. This in turn allows us to make the set $A$ and $B$ of size $\Omega(n)$. 
	
	For the lower bound of the node communication problem (c.f., Def.\ \ref{def:node_comm_problem}) in \NCC we can use a simplified form of Theorem \ref{thm:lower_bound_node_communication}, where we remove the dependency on $h$ (second argument of the min function) due to the lack of local communication. Then the rest follows from a reduction. Roughly speaking, by computing exact distance oracles and routing schemes nodes (who get all labels in advance ``for free'') can determine the existence edges that are sampled from a complete bipartite graph formed among nodes in $A$ (c.f., Theorem \ref{thm:lower_bound_unweighted})
	
	These edges are determined by a random bit vector $X$ of length $\Omega(n^2)$ (the number of edges of the complete bipartite graph). This solves the node communication problem on $X$, thus there must have been (roughly) $H(X) \in \Omega(n^2)$ bits communication from $A$ to $B$ whereas all nodes combined have a bandwidth of only $\tilO(n)$ bits per round. Not even all labels (``free information'') of size up to $cn$ bits for some sufficiently small $c > 0$ can help to communicate this information from $A$ to $B$. Note that for \NCC can also show roughly the same constant stretch factors that we achieved in the \hybrid model even for unweighted graphs.
\end{proof}

\shrt{}

\shrt{

\section{Deferred Proof of Section \ref{sec:node_communication_problem}}
\label{sec:node_communication_problem_proofs}

\begin{proof}[Proof of Lemma \ref{lem:reduce_comm_problems}]
	We will derive a protocol $\mathcal P$ that uses (i.e., simulates) algorithm $\mathcal A$ in order to solve the two-party communication problem. First we make a few assumptions about the initial knowledge of both parties in particular about the graph $G$ from the node communication problem, you can think of this information as hard coded into the instructions of $\mathcal P$. The important observation is that none of these assumptions will give Bob any knowledge about $X$.
	
	Specifically, we assume that Alice is given complete knowledge of the topology $G$ and inputs of all nodes in $G$ (in particular the state of $X$ and the source codes of all nodes specified by $\calA$). Bob is given the same for the subgraph induced by $V \setminus A$, which means that the state of $X$ remains unknown to Bob (c.f., Def.\ \ref{def:node_comm_problem}). To accommodate randomization of $\mathcal A$, both are given the same copy of a string of random bits (determined randomly and independently from $X$) that is sufficiently long to cover all ``coin flips'' used by any node in the execution of $\mathcal A$.
	
	Alice and Bob simulate the following nodes during the simulated execution of algorithm $\mathcal A$. For $i \in [h\!-\!1]$ let $V_i := \{v \in V \mid \hop(v,A) \leq i \}$ be the set of nodes at hop distance at most $i$ from $A$. Note that $A \subseteq V_i$ for all $i$. In round 0 of algorithm $\mathcal A$, Alice simulates all nodes in $A$ and Bob simulates all nodes in $V \setminus A$. However, in subsequent rounds $i > 0$, Alice simulates the larger set $A \cup V_i$ and Bob simulates the smaller set $B \cup V \setminus V_i$.
	
	Figuratively speaking, in round $i$ Bob will relinquish control of all nodes that are at hop distance $i$ from set $A$, to Alice. This means, in each round, every node is simulated \textit{either} by Alice \textit{or} by Bob.	
	We show that each party can simulate their nodes correctly with an induction on $i$. Initially ($i=0$), this is true as each party gets the necessary inputs of the nodes they simulate. Say we are at the beginning of round $i > 0$ and the simulation was correct so far. It suffices to show that both parties obtain all messages that are sent (in the \hybridpar{\infty}{\gamma} model) to the nodes they currently simulate.
	
	The communication taking place during execution of $\mathcal A$ in the \hybridpar{\infty}{\gamma} model is simulated as follows. If two nodes that are currently simulated by the \emph{same} party, say Alice, want to communicate, then this can be taken care as part of the internal simulation by Alice. If a node that is currently simulated (w.l.o.g.) by Bob wants to send a message over the \emph{global} network to some node that Alice simulates, then Bob sends that message directly to Alice as part of $\mathcal P$, and that message becomes part of the transcript.
	
	Now consider the case where a \textit{local} message is exchanged between some node $u$ simulated by Alice and some node $v$ simulated by Bob. Then in the subsequent round Alice will \textit{always} take control of $v$, as part of our simulation regime. Thus Alice can continue simulating $v$ correctly as she has all information to simulate all nodes all the time anyway (Alice is initially given all inputs of all nodes). Therefore it is \textit{not} required to exchange {any} \textit{local} messages across parties for the correct simulation.
	
	After $T$ simulated rounds, Bob, who simulates the set $B$ until the very end (as $T < h$), can derive the state of $X$ from the local information of $B$ with success probability at least $p$ (same as algorithm $\mathcal A$). Hence, using the global messages that were exchanged between Alice and Bob during the simulation of algorithm $\mathcal A$ we obtain a protocol $\mathcal P$ that solves the two party communication problem with probability $p$. Since total global communication is restricted by $n \cdot \gamma$ bits per round in the \hybridpar{\infty}{\gamma} model, Alice sends Bob at most $T \cdot n \cdot \gamma$ bits during the whole simulation.
\end{proof}

\begin{proof}[Proof of Corollary \ref{cor:lower_bound_general_amended}]
	As above, the node communication problem reduces to the communication problem between Alice and Bob, where Alice is now allowed to send $y$ bits to Bob in advance. Note that this still requires Alice to send $p \cdot H(X)-1-y$ remaining bits in expectation, as per Lemma \ref{lem:lower_bound_two_party}. 	
	The same contradiction as in Theorem \ref{thm:lower_bound_node_communication} can be derived as follows. Fewer rounds than stated in this lemma would imply that the transcript of global messages from Alice to Bob, would be shorter than $p \cdot H(X)-1-y$ bits (essentially by substituting $p \cdot H(X)-1$ for $p \cdot H(X)-1-y$ in the previous proof). Thus the transcript would be less than $p \cdot H(X)-1$ bits even when we add the $y$ ``free'' bits to the transcript.
\end{proof}

\section{Deferred Proofs of Section \ref{sec:lower_bounds_approx}}
\label{sec:lower_bounds_approx_proofs}

\begin{proof}[Proof of Lemma \ref{lem:source_target_distance_weighted}]
	Let $U := \{u_1, \dots ,u_k,v\}$ be the vertex cut that separates any $s_i$ from any $t_j$. The shortest \textit{simple} $s_i$-$t_j$-path that crosses $U$ via $v$ has length $w_2 + w_0 + h - 1$ \textit{independently} from $x_e$ (simple implies that a path can not ``turn around'' and go via $u_i$). 
	
	Consider the shortest $s_i$-$t_j$-path that does \textit{not} contain $v$. In the case $x_e=1$, i.e., $e = \{u_i,t_j\}$ exists in $\Gamma$, this $s_i$-$t_j$-path is forced to cross $U$ via $u_i$ and then goes directly to $t_j$ via $e$, and thus has length $w_2 + w_1 + h - 1$.
	
	Let us analyze the length of the $s_i$-$t_j$-path that does \textit{not} contain $v$ for the case $x_e=0$ (i.e., $e \notin E_\Gamma$). Let $G'$ be the subgraph that corresponds to $G$ after removing each edge $e' \in E$ with $x_{e'}=0$. Then that $s_i$-$t_j$-path has to traverse $G'$ to reach $t_j$.	
	The sub-path from $u_i$ to $t_j$ in $G'$ has to use at least $(\ell\!-\!1)$ edges, because otherwise $e = \{u_i,t_j\}$ would close a loop of less than $\ell$ edges in $G'$ (and thus also in $G$), contradicting the premise that $G$ has girth $\ell$. Thus, for $e \notin E_\Gamma$ any $s_i$-$t_j$-path that does not contain $v$ has length \textit{at least} $w_2+(\ell\!-\!1)w_1+h-1$.
	
	We sum up the cases. If $x_e = 1$, then the $s_i$-$t_j$-path \textit{not} containing $v$ of length $w_2 + w_1 + h - 1$ is shorter than the one via $v$ of length $w_2 + w_0 + h - 1$, since $w_1 < w_0$. If $x_e = 0$, then the $s_i$-$t_j$-path via $v$ of length $w_2 + w_0 + h - 1$ is shorter than the one not containing $v$ of length at least $w_2 + (\ell-1)w_1 + h - 1$ due to $w_0 < (\ell\!-\!1) w_1$.
\end{proof}

\begin{proof}[Proof of Lemma \ref{lem:approx_lower_bound_groundwork}]
	As $\mathcal A$ solves the node communication problem (Def.\ \ref{def:node_comm_problem}) it takes at least \smash{$\min\!\big(\frac{pH(X)-1-y}{n \cdot \gamma}, h\big)$} rounds by Corollary \ref{cor:lower_bound_general_amended}, where $H(X) \in \Theta (k^{1+\delta})$ (see property (2), further above), $p$ is the constant success probability and $y$ describes the ``free'' communication. 
	
	The arguments of \smash{$\min\!\big(\frac{pH(X)-1-y}{n \cdot \gamma}, h\big)$} behave inversely, since increasing the distance $h = hop(A,B)$ leaves only \smash{$k \in \Theta\big(\frac{n}{h}\big)$} nodes for the graph $G$, which decreases $H(X) \in \Theta (k^{1+\delta})$. So in order to maximize the number of rounds given by the min function, we solve the equation \smash{$\frac{pH(X)-1-y}{n \cdot \gamma} = h$} subject to $k \cdot h = \Theta(n)$. Slashing constants and neglecting $y$ for now, this simplifies as follows
	\[
	\Theta \big( \tfrac{k^{1+\delta}}{n\cdot \gamma} \big) = \Theta (h), \quad \text{subject to} \quad k \cdot h = \Theta(n).
	\] 
	The solution is \smash{$k = \Theta\big(n^{\frac{2}{2+\delta}}\cdot \gamma^{\frac{1}{2+\delta}}\big)$} and \smash{$h = \Theta\big(n^{\frac{\delta}{2+\delta}}/ \gamma^{\frac{1}{2+\delta}}\big)$} (which the willing reader may verify by inserting), resulting in a lower bound of \smash{$\Omega(n^{\frac{\delta}{2+\delta}} / \gamma^{\frac{1}{2+\delta}}) = \Omega\big((\tfrac{n^\delta}{\gamma})^{\frac{1}{2+\delta}}\big)$} rounds.
	
	When we factor the free communication of $y = c \cdot k^{1+\delta}$ bits back into the equation \smash{$\frac{pH(X)-1-y}{n \cdot \gamma} = h$}, then there are no asymptotic changes to the outcome of our calculations as long as we choose the constant $c$ such that \smash{$c \cdot k^{1+\delta} \leq \frac{pH(X)-1}{2}$} (for all $n$ bigger than some constant $n_0$).
\end{proof}

\begin{proof}[Proof of Lemma \ref{lem:lower_bound_distance_oracle}]
	Set $w_2 = 1$ (this weight is only needed later for stateful routing scheme lower bounds). To make the idea described above work, we have to make the unweighted edges of $\Gamma$ (the $h$-hop $s_i$-$u_i$-paths) insignificant for the approximation ratio by scaling the weights $w_1,w_0$ by a large factor $t$.
	For that purpose we introduce the parameter $t>0$, which is specified later. 	
	We choose $w_1 = t$ and $w_0 =(\ell\!-\!1\!-\!\frac{\varepsilon}{2})\cdot t$. In particular, this means $w_1 < w_0 < (\ell-1) w_1$ (the precondition of Lemma \ref{lem:source_target_distance_weighted}).
	
	Let $e = \{u_i,t_j\} \in E$. Let $d_1 = d(s_i,t_j)$ for the case $x_e = 1$ and $d_0 = d(s_i,t_j)$ for the case $x_e = 0$. By Lemma \ref{lem:source_target_distance_weighted}, we know that $d_1 = t \!+\! h$ and \smash{$d_0 = (\ell\!-\!1\!-\!\frac{\varepsilon}{2})\cdot t + h$}.		
	Let $\mathcal A$ be an approximation algorithm for the distance oracle problem with stretch $\alpha_\ell = \ell\!-\!1\!-\!\varepsilon$. 
	
	For the cases $x_e = 1$ and $x_e = 0$, respectively, let $\tilde d_1$ and $\tilde d_0$ be distance approximations of $d(s_i,t_j)$ with stretch $\alpha_\ell$ that $s_i$ determines with its local table and the label of $t_j$ (which we computed with $\mathcal A$). 
	Note that our claims about $\tilde d_0, \tilde d_1$ will only depend on $x_e$ and are independent from $x_{e'}$ of other edges $e'\in E \setminus \{e\}$ (even though the \textit{exact} value of $\tilde d_0, \tilde d_1$ might depend on the $x_{e'}$).
	
	The goal is to show 
	$\tilde d_1 < d \leq \tilde d_0$ for some constant $d > 0$, which enables $s_i$ to distinguish $x_e=1$ from $x_e=0$ from its approximation of $d(s_i,t_j)$.
	We know that $\tilde d_1 \leq \alpha_\ell \cdot d_1 = (\ell\!-\!1\!-\!\varepsilon)(t \!+\! h)$. We also have $\tilde d_0 \geq d_0$ since our approximations are supposed to be one-sided.
	Then 
	\begin{align*}
	\tilde d_1 & \leq (\ell\!-\!1\!-\!\varepsilon)(t \!+\! h)\\ 		
	& = (\ell\!-\!1\!-\!\tfrac{\varepsilon}{2})\cdot t + h + (\ell\!-\!2\!-\!\varepsilon) \cdot h -\tfrac{\varepsilon}{2}\cdot t \tag*{\small\text{\textit{expand}}}\\		
	& < (\ell\!-\!1\!-\!\tfrac{\varepsilon}{2})\cdot t + h \tag*{\small\text{\textit{for large enough} $t$}}\\
	& = d_0 \leq \tilde d_0.
	\end{align*}
	
	The strict inequality is obtained by choosing $t > \frac{2(\ell-2-\varepsilon)}{\varepsilon} \cdot h \in \Theta(\frac{\ell h}{\varepsilon})$. Note that edge weights remain polynomial in $n$ with this choice of $t$ since $\ell, h \in O(n)$ and $\varepsilon$ is constant. 	
	So we get $\tilde d_1 < d_0 \leq \tilde d_0$, which implies the following. Let \smash{$\tilde d(s_i,t_j)$} be the distance estimate that $s_i$ actually outputs. Then it is $x_e=0$ if \smash{$\tilde d(s_i,t_j) < d_0$}, else it is $x_e=1$. Hence the nodes $B := \{s_1, \dots, s_k\}$ collectively learn $X$ and thus solve the node communication problem, which takes \smash{$\Omega\big((\tfrac{n^\delta}{\gamma})^{\frac{1}{2+\delta}}\big)$} rounds by Lemma \ref{lem:approx_lower_bound_groundwork}.
	
	Two things remain to be mentioned. First, we can assume that each $s_i$ also has advance knowledge of $d_0$, as this does not carry any information about $X$ to the nodes $B = \{s_1, \dots, s_k\}$ and therefore does not make the node communication problem easier. 
	
	Second, the nodes $s_i$ can only produce the distance estimations \smash{$\tilde d(s_i,t_j)$} when they are also provided with the labels $\lambda(t_1), \dots, \lambda(t_k)$.
	Here we assume that the nodes $s_i$ get these labels in advance as part of the contingent of ``free'' communication that we budgeted for in Lemma \ref{lem:approx_lower_bound_groundwork} and which does not make the node communication problem asymptotically easier. 
	
	In particular let $c_1>0$ be the constant from Lemma \ref{lem:approx_lower_bound_groundwork}, such that we are allowed \smash{$y = c_1 k^{1+\delta}$} bits of free communication in total. This leaves \smash{$c_1 k^{\delta}$} bits for each of the $k$ labels $\lambda(t_1), \dots, \lambda(t_k)$. In the proof of Lemma \ref{lem:approx_lower_bound_groundwork} we chose \smash{$k = c_2 \cdot n^{\frac{2}{2+\delta}}\cdot \gamma^{\frac{1}{2+\delta}}$} (for some constant $c_2 > 0$), resulting in a label size of at most \smash{$c_1 c_2 \cdot n^{\frac{2\delta}{2+\delta}}\cdot \gamma^{\frac{\delta}{2+\delta}}$}.
\end{proof}	

\begin{proof}[Proof of Lemma \ref{lem:lower_bound_stateless_routing}]
	We set $w_2 = 1$ (this weight only plays a role later \textit{stateful} routing lower bound) and $w_1 < w_0 < (\ell \!- \! 1)w_1$ (more precise values are determined further below). Let $s_i, t_j$ be a source-target pair of $\Gamma$ where $e := \{s_i,t_j\} \in E$ is part of $G$ (but not necessarily $\Gamma$). For the case $x_e = 1$ we define $d_1 := d(s_i,t_j) = w_1\!+\!h$ and $d_0 := d(s_i,t_j) = w_0\!+\!h$ for the case $x_e = 0$ (c.f., Lemma \ref{lem:source_target_distance_weighted}).
	
	Let $\mathcal A$ be an algorithm solving the \textit{stateless} routing problem with the claimed approximation ratio (with const.\ probability). Let $P$ be the \textit{simple} $s_i$-$t_j$-path induced by the \textit{stateless} routing scheme computed by $\mathcal A$. Recall that $P$ must be simple as otherwise a packet that is oblivious to the prior routing decisions would be trapped in a loop. Our aim is that $s_i$ can decide whether $x_e =0$ or $x_e =1$ from the next routing node on $P$. That is, we want that $P$ contains $v$ if and only if $x_e = 0$. However, there are a few options to obtain the $s_i$-$t_j$-path $P$ that do, in certain cases, not abide by this requirement. We enumerate these in the following.
	
	Assume $P$ does not contain $e$. Then option (1) is to go the left ``lane'' via $v$ and then use the direct blue edge $\{v', t_j\}$ to get to $t_j$, i.e., the path of length $d_0$. If $P$ goes from $s_i$ directly to $u_i$, then $P$ can be completed into a path to $t_j$ by (2) using only edges that are also part of $G$ (but not $e$), i.e., only red edges in Figure \ref{fig:lower_bound_enhanced}. Note that option (2) must include at least $\ell \! - \! 1$ red edges due to the girth $\ell$ of $G$. Option (3) is where $P$ goes to $u_i$ first uses any red edge $\{u_i,t_p\}$ ($p \neq j$) and then the two blue edges $\{t_p, v'\}, \{v', t_j\}$ to reach $t_j$. Note that for the case $x_e = 0$ all other $s_i$-$t_j$-paths are either strictly longer than option (1),(2),(3) or not simple.
	The distances of the three paths (1),(2),(3) are at least (recall $w_2 = 1$):	
	\begin{align*}
	d^{(1)} & := w_0 + h \;(= d_0)\\
	d^{(2)} & := (\ell \! - \! 1)w_1 + h\\
	d^{(3)} & := w_1 + 2w_0 + h
	\end{align*}
	To enforce $v \notin P$ in case $x_e = 1$, we require that path (1) is unfeasible, i.e., exceeds the allowed distance $\alpha_\ell d_1$. Furthermore, to enforce $v \in P$ in case $x_e = 0$, we require that the paths (2) and (3) exceed the allowed distance $\alpha_\ell d_0$. From this we obtain the following conditions.
	\begin{align*}
	\alpha_{\ell} \cdot d_1 & \stackrel{!}{<} d^{(1)} = w_0 + h\tag{1}\\
	\alpha_{\ell} \cdot d_0 & \stackrel{!}{<} d^{(2)} = (\ell \! - \! 1)w_1 + h \tag{2} \\
	\alpha_{\ell} \cdot d_0 & \stackrel{!}{<} d^{(3)} = w_1 + 2w_0 + h \tag{3}
	\end{align*}
	The remaining part of the proof is merely technical, we need to maximize $\alpha_\ell$ under the above constraints. (Afterwards, the rest follows from having solved the node communication problem as in the proof of Lemma \ref{lem:lower_bound_distance_oracle}). Set $w_1 := t < w_0$ for some yet unspecified variable $t > 0$. Further, we set $\alpha := \frac{w_0}{t} - \eps$ and show that this fulfills Equation (1):
	\[
	\alpha_{\ell} \cdot d_1 = \big(\tfrac{w_0}{t} - \eps\big)(t+h) = w_0 + \tfrac{h w_0}{t} - \eps(t+h) < w_0 + h.
	\]
	For Equation (2) we obtain:
	\begin{align*}
	& \alpha_{\ell} \cdot d_0 = \big(\tfrac{w_0}{t} \!-\! \eps\big)(w_0\!+\!h)  < (\ell \! - \! 1)t + h\\
	\Longleftrightarrow \quad & \tfrac{w_0^2}{t} + \tfrac{w_0}{t}(h-\eps t)-\eps h - h < (\ell \! - \! 1)t\\
	\Longleftrightarrow \quad & w_0^2 + \underbrace{w_0(h-\eps t)}_{<0, \text{ for } t > h/\eps} \underbrace{-\eps ht - ht}_{<0} < (\ell \! - \! 1)t^2\\
	{\Longleftarrow} \quad & w_0^2 \leq (\ell \! - \! 1)t^2\\
	{\Longleftrightarrow} \quad & w_0 \leq t \cdot \sqrt{\ell \! - \! 1}.
	\end{align*}
	Now let us turn to Equation (3)
	\begin{align*}
	& \alpha_{\ell} \cdot d_0 = \big(\tfrac{w_0}{t} \!-\! \eps\big)(w_0\!+\!h)  < t + 2w_0 + h\\
	\Longleftrightarrow \quad & \tfrac{w_0^2}{t} + \tfrac{w_0}{t}(h-\eps t)-\eps h < t + 2w_0 + h\\
	\Longleftrightarrow \quad & w_0^2 + w_0(\underbrace{h-\eps t}_{<0}-2t) \underbrace{-\eps ht - ht}_{<0} - t^2 < 0\\
	\Longleftarrow \quad & w_0^2 -2tw_0 - t^2\leq 0 \\
	\Longleftarrow \quad & w_0 \leq t \cdot (1\!+\!\sqrt 2).
	\end{align*}
	Therefore, Equations (2),(3) are fulfilled if $w_0 \leq t \cdot \sqrt{\ell \! - \! 1}$ and $w_0 \leq t \cdot (1\!+\!\sqrt 2)$ (and \smash{$t > \frac{h}{\eps}$}, which can be chosen freely). This implies that the stretch $\alpha_\ell = \frac{w_0}{t} - \eps$ must satisfy $\alpha_\ell\leq \sqrt{\ell \! - \! 1} - \eps$ and $\alpha_\ell \leq 1\!+\!\sqrt 2 \!-\! \eps$, whereas the former condition for $\alpha_\ell$ dominates the latter if and only if $\ell \leq 8$.
	
	With this choice of $w_0, \alpha_\ell$, the first node that a packet from $s_i$ with destination $t_j$ is routed to is the node $v$ if and only if $x_e = 0$. Hence, the nodes in $B$ collectively learn $X$ from the information provided by algorithm $\calA$ and the labels of the nodes $t_j$, which therefore solves the node communication problem. The runtime and size of the labels then follows the same way as in the proof of Lemma \ref{lem:lower_bound_distance_oracle}.
\end{proof}	

\begin{proof}[Proof of Lemma \ref{lem:lower_bound_stateful_routing}]
	
	The beginning of the proof is similar to the one of Lemma \ref{lem:lower_bound_stateless_routing}.	
	Consider algorithm $\mathcal A$ that solves the \textit{stateful} routing problem with the claimed approximation ratio and constant probability. Let $s_i, t_j \in V_\Gamma$ with $e := \{s_i,t_j\} \in E$ and define $d_1 := d(s_i,t_j) = w_1\!+\!w_2\!+\!h\!-\!1$ and $d_0 := d(s_i,t_j) = w_0\!+\!w_2\!+\!h\!-\!1$ for the cases $x_e = 1$, $x_e = 0$ respectively (we will ensure $w_1 < w_0 < (\ell\!-\!1)w_1$ so that Lemma \ref{lem:source_target_distance_weighted} applies). 
	
	Let $P$ be the $s_i$-$t_j$-path induced by the \textit{stateful} routing scheme computed by $\mathcal A$. Similar to before, our goal is to show that the first node on $P$ is $v$ if and only if $x_e = 0$, so that $s_i$ learns $x_e$ from its routing decision. As before, we aim to prohibit all routing options (i.e., make them break the stretch guarantee) that do not abide by this requirement. We will extend our list of options for $P$ which are in some cases undesirable from the previous proof.
	
	We have the loop-less routing options (1),(2),(3) from before (c.f., proof of Lemma \ref{lem:lower_bound_stateless_routing}), whose lengths change by an additive term $w_2-1$ due to the introduction of weight $w_2$ (orange edges in Figure \ref{fig:lower_bound_enhanced}). Note that any path (possibly with loops) that contains $s_i$ just once is at least as long as one of the options (1)-(3) (in their respective cases) so prohibiting the latter prohibits the former.
	
	
	The problem that arises is from visiting $s_i$ at least twice is that it can mislead $s_i$ by first going to $v$ even though $x_e = 1$ or vice versa. 
	Assume the case that $e$ is not in $\Gamma$ ($x_e = 0$), then we need to prohibit routing option (4) that visits the first node on the path towards $u_i$, returns to $s_i$, travels directly to $v'$ and uses the blue edge to $t_j$.		
	Conversely, assuming the case that $e$ is present in $\Gamma$ ($x_e = 1$), $P$ we need to prohibit option (5), which visits $v$ first but then returns to $s_i$, travels to $u_i$ and uses $e$ to reach $t_j$. Note that all paths that contain $s_i$ at least twice, are at least as long as one of the paths (4),(5) (in their respective cases). The described routing paths (1)-(5) have the following respective lengths:
	\begin{align*}
	d^{(1)} & := w_0 + w_2 + h -1\\
	d^{(2)} & := (\ell \! - \! 1)w_1 + w_2 + h -1\\
	d^{(3)} & := 2w_0 + w_1 + w_2 + h -1\\
	d^{(4)} & := w_0 + 3w_2 + h -1\\
	d^{(5)} & := w_1 + 3w_2 + h -1
	\end{align*}	
	
	We obtain inequalities similar as before. For the reasoning of inequalities (1)-(3) consider the proof of Lemma \ref{lem:lower_bound_stateless_routing}. Observe that if $P$ would follow path options (4),(5) it could make $s_i$ believe $x_e = 1$ or $x_e = 0$ respectively, even though the opposite is true. We require that path options (4) and (5) exceed the stretch $\alpha_\ell$ in the cases in which they are undesirable, that is, $x_e = 0$ and $x_e = 1$, respectively. 
	\begin{align*}
	\alpha_{\ell} \cdot d_1 & \stackrel{!}{<} d^{(1)} = w_0 + w_2 + h -1\tag{1}\\
	\alpha_{\ell} \cdot d_0 & \stackrel{!}{<} d^{(2)} = (\ell \! - \! 1)w_1 + w_2 + h -1 \tag{2} \\
	\alpha_{\ell} \cdot d_0 & \stackrel{!}{<} d^{(3)} = 2w_0 + w_1 + w_2 + h -1 \tag{3} \\
	\alpha_{\ell} \cdot d_0 & \stackrel{!}{<} d^{(4)} = w_0 + 3w_2 + h -1 \tag{4} \\
	\alpha_{\ell} \cdot d_1 & \stackrel{!}{<} d^{(5)} = w_1 + 3w_2 + h -1 \tag{5}
	\end{align*}
	This leaves us with an optimization problem where $\alpha_{\ell}$ is to be maximized in the domain $1 \leq w_1 \leq w_0,w_2$, subject to conditions (1)-(5). This is admittedly a bit tedious, in particular since the optimization problem has a different result for girth $\ell \in \{4,6,8,10\}$ (we gain no improvement for larger $\ell$). We do not reproduce all necessary calculations of the optimization in detail, instead we give some explanations to make our results (given in Table \ref{tbl:lower_bounds_stateful_opt_values}) reproducible. 
	
	One obstacle is the presence of edges with weight 1 in $\Gamma$ forming the path of length $h-1$, which is not yet fixed and therefore prohibits solving the optimization problem directly. Our strategy is to first relax our problem to a simplified one with zero-weight edges. That is, we set all edges, except those with weight $w_0,w_1,w_2$, to zero and adapt the distances $d_0, d_1, d^{(1)}, \dots,  d^{(5)}$ accordingly (essentially slashing $h\!-\!1$ in those distances). We also set $w_1 = 1$. This eliminates $w_1,h$ from conditions (1)-(5) and the resulting maximization problem(s) can be solved directly.
	
	From the results of the simplified optimization problems (for girths $\ell \in \{4,6,8,10\}$), we obtain almost optimal solutions for the general optimization problem with non-zero weights as follows. We multiply the values of $w_0,w_1,w_2$ with a (sufficiently large) parameter $t \geq 1$ and subtract an $\eps>0$ from $\alpha_\ell$ to accommodate a (small) slack that is required in the general inequalities. The value of $t$ depends on the slack $\eps$ in $\alpha_\ell$ that is given but is generally polynomial in $n$ as long as $\eps$ is constant.
	The results obtained using this procedure are given in Table \ref{tbl:lower_bounds_stateful_opt_values}.
	
	\begin{table}[h]
		\begin{center}
		\begin{tabular}{p{5mm}>{\centering\arraybackslash}p{14mm}>{\centering\arraybackslash}p{14mm}>{\centering\arraybackslash}p{6mm}>{\centering\arraybackslash}p{14mm}}
			\toprule
			$\ell$ & $\alpha_\ell$ & $w_0$ & $w_1$ & $w_2$\\
			\midrule
			4 & $\sqrt 2 - \eps$ & $2t\sqrt 2 \!-\!t$ & $t$ & $t$\\[1.5mm]
			6 & $\frac{5}{3} - \eps$ & $\frac{5t}{2}$ & $t$ & $\frac{5t}{4}$ \\[1.5mm]
			8 & $\tfrac{7}{4} - \eps$ & $\frac{35t}{11}$ & $t$ & $\frac{21t}{11}$ \\[1.5mm]
			10 & $\frac{3+\sqrt{17}}{4} - \eps$ & $\frac{3+\sqrt{17}}{2}t$ & $t$ & $\frac{5+\sqrt{17}}{4}t$\\
			\bottomrule
		\end{tabular}
		\end{center}
		\caption{\lng{\boldmath} Results of maximizing $\alpha_\ell$, s.t., conditions (1)-(5) for $\ell \in \{4,6,8,10\}$ (and scaling with $t$ and introducing slack $\eps$).}
		\label{tbl:lower_bounds_stateful_opt_values}
	\end{table}
	
	The most relevant parameter in Table \ref{tbl:lower_bounds_stateful_opt_values} are certainly the values for $\alpha_\ell$ for $\ell \in \{4,6,8,10\}$, whereas we also provide the weights $w_0,w_1,w_2$ for reproducibility. Showing that the given parameters do in fact satisfy conditions (1)-(5) for any constant $\eps > 0$ and some choice of $t$, is a repetitive task. We show this once for $\ell=10$ and condition (1) (which is almost tight in this case), the other cases can be repeated analogously.
	
	\begin{align*}
	\alpha_{\ell} \cdot d_1 &  = \big(\tfrac{3+\sqrt{17}}{4} \big)d_1 - \eps d_1\\
	& = \big(\tfrac{3+\sqrt{17}}{4} \big)(w_1 + w_2 + h -1) - \eps d_1\\
	& = \big(\tfrac{3+\sqrt{17}}{4} \big)(w_1 + w_2) + h \!-\!1 + \smash{\underbrace{\big(\tfrac{\sqrt{17}-1}{4} \big)(h\!-\!1) - \eps d_1 }_{<0,\text{ for large }t}}\\
	& < \big(\tfrac{3+\sqrt{17}}{4} \big)\big(\tfrac{9+\sqrt{17}}{4}\cdot t\big) + h \!-\!1\\
	& = \big(\tfrac{11+3\sqrt{17}}{4}\big) \cdot t + h \!-\!1\\
	& = \big(\tfrac{3+\sqrt{17}}{2} + \tfrac{5+\sqrt{17}}{4}\big) \cdot t + h \!-\!1\\
	& = w_0 + w_2 + h \!-\!1 = \smash{d^{(1)}}.
	\end{align*}
	Note that, since $d_1$ grows linear in $t$ the inequality in the fourth line holds for \smash{$t > \big(\tfrac{\sqrt{17}-1}{4\eps} \big)(h\!-\!1) \in \bigO(n)$}.
	
	Finally, by enforcing conditions (1)-(5) with the appropriate parameters given in Table \ref{tbl:lower_bounds_stateful_opt_values}, we can guarantee that the first node that a packet from $s_i$ to $t_j$ is first routed to $v$ if and only if $x_e = 0$. From this the nodes in $B$ collectively learn $X = (x_e)_{e \in E}$ thus solving the node communication problem. The runtime and size of the labels follows as in the proof of Lemma \ref{lem:lower_bound_distance_oracle}.
\end{proof}
}

\bibliography{ref/hybrid_routing_refs}

\end{document}